\documentclass[onecolumn, 11 point, journal]{IEEEtran}
    \usepackage{graphicx}
    \usepackage{amsmath}
    \usepackage{amsthm}
    \usepackage{amssymb}
    \usepackage{cite}

    \usepackage{blindtext}

    \usepackage{afterpage}
    \usepackage{enumitem}
    \usepackage[normalem]{ulem}
    \newtheorem{theorem}{Theorem}
    \newtheorem{lemma}{Lemma}

    \newtheorem{definition}{Definition}
    \newtheorem{claim}[theorem]{\bf Claim}
    
    \newtheorem{example}{Example}
    \newtheorem{construction}{Construction}
    \let\olddefinition\definition
    \renewcommand{\definition}{\olddefinition\normalfont}
    \let\oldexample\example
    \renewcommand{\example}{\oldexample\normalfont}
    \newtheorem{remark}{Remark}
    \let\oldremark\remark
    \renewcommand{\remark}{\oldremark\normalfont}
    \newtheorem{assumption}{Assumption}
    \let\oldassumption\assumption
    \renewcommand{\assumption}{\oldassumption\normalfont}

    \usepackage{tabu}
    \usepackage{amsmath}
    \usepackage{graphicx}
    \usepackage[colorinlistoftodos]{todonotes}
    \usepackage[colorlinks=true, allcolors=blue]{hyperref}

    \IEEEoverridecommandlockouts

    \title{On the Optimal Recovery Threshold of Coded Matrix Multiplication}
    \author{ Sanghamitra Dutta$^{*}$, Mohammad Fahim$^{*}$, Farzin Haddadpour$^{*}$, Haewon Jeong$^{*},$ Viveck Cadambe, Pulkit Grover
    \thanks{
    $^{*}$ author ordering in alphabetical order}
    \thanks{
    S. Dutta, H. Jeong and P. Grover are with the Department of Electrical and Computer Engineering, Carnegie Mellon University, Pittsburgh, PA-15213. M. Fahim, F. Haddadpour and V. Cadambe are with the Department of Electrical Engineering, Pennsylvania State University, University Park, PA 16802. } \thanks{This work was presented in part at the Annual Allerton Conference on Communication, Control, and Computing (Allerton) in October 2017 \cite{allerton17}.}
    }

\begin{document}
    \bstctlcite{IEEEexample:BSTcontrol}
    \maketitle
    \begin{abstract}
    We provide novel coded computation strategies for distributed matrix-matrix products that outperform the recent ``Polynomial code'' constructions in recovery threshold, \textit{i.e.}, the required number of successful workers. When $m$-th fraction of each matrix can be stored  in each worker node, Polynomial codes require $m^2$ successful workers, while our MatDot codes only require $2m-1$ successful workers, albeit at a higher communication cost from each worker to the fusion node. We also provide a systematic construction of MatDot codes. Further, we propose ``PolyDot'' coding that interpolates between Polynomial codes and MatDot codes to trade-off communication cost and recovery threshold. Finally, we demonstrate a coding technique for multiplying $n$ matrices ($n \geq 3$) by applying MatDot and PolyDot coding ideas.
    \end{abstract}

    \section{Introduction}
As the era of Big Data advances, massive parallelization has emerged as a natural approach to overcome limitations imposed by saturation of Moore's law (and thereby of single processor compute speeds). However, massive parallelization leads to computational bottlenecks due to faulty nodes and stragglers~\cite{dean2013tail}. Stragglers refer to a few slow or delay-prone processors that can bottleneck the entire computation because one has to wait for all the parallel nodes to finish. The issue of straggling~\cite{dean2013tail} and faulty nodes has been a topic of active interest in the emerging area of ``coded computation'' with several interesting works, e.g.~\cite{lee2018speeding,li2015coded,tandon2016gradient,polynomialcodes,joshi2014delay,gauristraggler,gauriefficient,dutta2016short,azian2017consensus,YaoqingAllerton16,Yaoqing2017ITTrans,Yang_ISTC_16,Salman1,Salman2,Salman3,Salman4,GC2,GC3,Emina1,Emina2,Virtualization,heterogeneousclusters,GC4,Suhas1,Suhas2,NIPS17Yaoqing,Ramtin1,multicore_setups,yu2017fft,jeongFFT,baharav2018straggler,suh2017matrix,mallick2018rateless, wang2018coded, wang2018fundamental,severinson2017block,ye2018communication}. Coded computation not only advances on coding approaches in classical works in Algorithm-Based Fault Tolerance (ABFT)~\cite{Huang_TC_84,faultbook}, but also provides novel analyses of required computation time (e.g. expected time~\cite{lee2018speeding} and deadline exponents~\cite{SanghamitraISIT2017}). Perhaps most importantly, it brings an information-theoretic lens to the problem by examining fundamental limits and comparing them with existing strategies. A broader survey of results and techniques of coded computation is provided in~\cite{NewsletterPaper}.

In this paper, we focus on the problem of coded matrix multiplication. Matrix multiplication is central to many modern computing applications, including machine learning and scientific computing. There is a lot of interest in classical ABFT literature (starting from~\cite{Huang_TC_84,faultbook}) and more recently in coded computation literature (e.g.~\cite{ProductCodes,polynomialcodes}) to make matrix multiplications resilient to faults and delays. In particular, Yu, Maddah-Ali, and Avestimehr~\cite{polynomialcodes} provide novel coded matrix-multiplication constructions called \emph{Polynomial codes} that outperform classical work from ABFT literature in terms of the \emph{recovery threshold}, the minimum number of successful (non-delayed, non-faulty) processing nodes required for completing the computation.

In this work, we consider the standard setup used in \cite{polynomialcodes,ProductCodes} with $P$ worker nodes that perform the computation in a distributed manner and a master node that helps coordinate the computation by performing some low complexity pre-processing on the inputs, distributing the inputs to the workers, and aggregates the results of the workers possibly performing some low complexity post-processing.\footnote{In this paper, we introduce a new type of node, \emph{``a fusion node''}, and delegate master node's result aggregation and post-processing function to the fusion node. Hence, a master node is only responsible for pre-processing and job distribution as a fusion node performs aggregating results and post-processing. However, this is only a conceptual separation that makes our explanation easier throughout the paper. One can think of a master node and a fusion node as one physical machine.} We propose MatDot codes that advance on existing constructions in scaling sense under the setup: when $m$-th fraction of each matrix can be stored in each worker node, Polynomial codes have the recovery threshold of $m^2$, while the recovery threshold of MatDot is only $2m-1$. However, as we note in Section~\ref{sec:complexity}, this comes at an increased per-worker communication cost. We also propose PolyDot codes that interpolates between MatDot and Polynomial code constructions in terms of recovery thresholds and communication costs. %\sout{We begin with some motivating examples that show our constructions and a summary of results.} \textcolor{red}{SD: we are giving motivating examples much later. I think we do not need this line here?}

Our main contributions in this work are as follows:
\begin{itemize}
    \item We present our system model in Section \ref{sec:sysmod}, and describe MatDot codes in Section \ref{sec:main}. While Polynomial codes have a recovery threshold of $\Theta(m^{2}),$ MatDot codes have a recovery threshold of $\Theta(m)$ when each node stores only $m$-th fraction of each matrix multiplicand.
    \item We present a \emph{systematic} version of MatDot codes, where the  operations of the first $m$ worker nodes may be viewed as multiplication in uncoded form, in Section \ref{sec:syscod}.
    \item In Section \ref{sec:polydot}, we propose ``PolyDot codes'',  a unified view of MatDot and Polynomial codes that leads to a trade-off between recovery threshold and communication costs.
    \item In Section \ref{sec:multiple_matrices}, we apply the constructions of Section \ref{sec:main} to study coded computation for multiplying more than two matrices.
\end{itemize}
We note that following the publication of an initial version of this paper~\cite{allerton17}, the works of Yu, Maddah-Ali, and Avestimehr~\cite{entangledpolycodes} and Dutta, Bai, Jeong, Low and Grover~\cite{DNNPaperISIT} obtained constructions that outperform PolyDot codes in trade-offs between communication cost and recovery threshold (although MatDot codes continue to have the smallest recovery threshold for given storage constraints). Importantly, Yu et al.~\cite{entangledpolycodes} also provide interesting converse results that show the optimality of MatDot codes. \\

%\textcolor{red}{Version 2 for this same para (SD and HJ): We note that \textcolor{blue}{subsequent} works of Yu, Maddah-Ali, and Avestimehr~\cite{entangledpolycodes} and Dutta, Bai, Grover, and Low~\cite{DNNPaperISIT} obtain constructions that outperform the PolyDot construction proposed in this paper in trade-offs between communication and recovery threshold (although MatDot codes continue to have the smallest recovery threshold for given storage constraints). Importantly, Yu et al.~\cite{entangledpolycodes} also provide fundamental limits demonstrating approximate optimality of these improved strategies. For the communication constraints considered in this paper, subsequent works \cite{entangledpolycodes,DNNPaperISIT} are able to achieve a recovery threshold that is lower by a factor of $2$.}\\

%\textcolor{red}{Comment (Maybe not write this  directly): We realized that one more thing that they are claiming credit for in the paper is that since there are $p^2mn$ terms, but the recovery threshold is of $O(pmn)$, i.e. reduction by a factor of $p$. In scaling sense, we also have this reduction, i.e. $t^2s^2$ terms but we achieve $t^2(2s-1)$ which is also in scaling sense a factor of $s$ lower. This also holds for Mat-Dot. In fact, Mat-Dot and Poly-Dot are the first codes that explores this reduction.}

\section{System model and problem statement} \label{sec:sysmod}

\subsection{System model}\label{sec:model}

The system, illustrated in Fig.~\ref{fig:system}, consists of three different types of nodes, a master node, multiple worker nodes, and a fusion node. These are defined  more formally below.
%\begin{definition}[Workers or Worker nodes] Worker nodes receive inputs from the master node and perform predetermined computations on these inputs in parallel. These nodes have storage limitations.
%\end{definition}
\begin{figure}[h]
\includegraphics[width=0.48\textwidth]{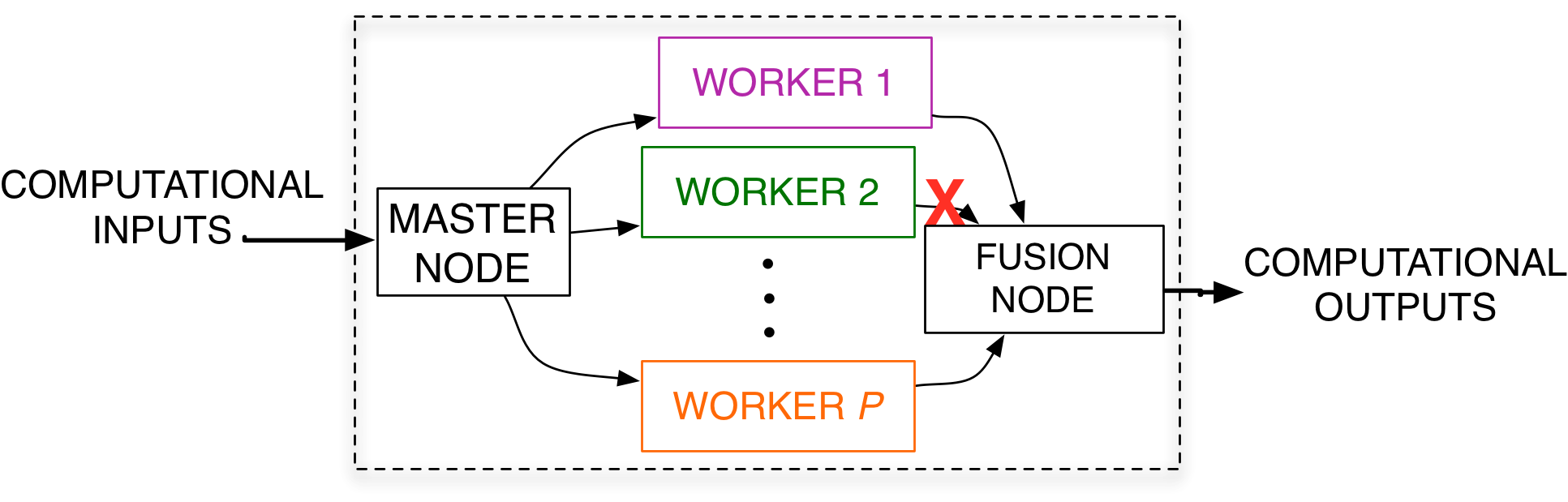}
\centering
\caption{The computational system: Master node receives the computational inputs and sends appropriate tasks to the workers. The workers are prone to faults and delays. The fusion node aggregates the partial computational outputs from the workers and finally produces the desired computational outputs.}
\label{fig:system}
\end{figure}

\begin{definition}[Computational system]\label{def:compsys}
A computational system consists of the following:
\begin{enumerate}
\item[(i)] a \textit{master node} that receives computational inputs, pre-processes them (e.g., encoding), and distributes them to the worker nodes
\item[(ii)] memory-constrained \textit{worker nodes} that perform pre-determined computations on their respective inputs in parallel
\item[(iii)] a \textit{fusion node} that receives outputs from successful worker nodes and performs post-processing (e.g.,decoding) to recover the final computation output.
\end{enumerate}
\end{definition}
For practical utility, it is important to have the amount of processing that the worker nodes perform to be much smaller than the processing at the master and the fusion node.

We assume that any worker node can fail to complete its computation because of faults or delays. Thus, we define a subset of all workers  as the \emph{``successful workers.''}
\begin{definition}[Successful workers]
Workers that finish their computation task successfully and send their output to the fusion node  are called successful {workers.}
\end{definition}

\begin{definition}[Successful computation]
If the computational system on receiving the inputs produces the correct computational output, the computation is said to be successful.
\end{definition}
\begin{definition}[Recovery threshold]\label{def:recovery}
The recovery threshold is the worst-case\footnote{The worst-case here is over all possible configurations of $k$ successful workers.} minimum number of successful workers required by the fusion node to complete the computation successfully.
\end{definition}
We will denote the total number of worker nodes by $P$, and the recovery threshold by $k$. We will be using the term ``row-block'' to denote the submatrices formed when we split a matrix $\mathbf{A}$ horizontally as follows: $\mathbf{A} = \begin{bmatrix} \mathbf{A}_0 \\
\mathbf{A}_1
\end{bmatrix}
$. Similarly, we will be using the term ``column-block'' to denote the submatrices formed when we split a matrix vertically into submatrices as follows: $\mathbf{A} = \begin{bmatrix} \mathbf{A}_0 \ \
\mathbf{A}_1
\end{bmatrix}$.

%\begin{definition}[Successful computation]\label{def:compsys}
%\end{definition}
%%%%%%%%%%%%%%%%%%%%%%%%%%%%%%%%
\subsection{Problem statement}\label{sec:problem}

We are required to compute the multiplication of two square matrices $\mathbf{A}, \mathbf{B} \in \mathbb{F}$ ($|\mathbb{F}| > P$), \textit{i.e.}, $\mathbf{AB}$ using the computational system specified in Section~\ref{sec:model}. Both the matrices are of dimension $N \times N$.  Each worker can receive at most  $2N^2/m$ symbols from the master node, where each symbol is an element of $\mathbb{F}$. For the simplicity, we assume that $m$ divides $N$ and a worker node receives $N^2/m$ symbols from $\mathbf{A}$ and $\mathbf{B}$ each\footnote{We only consider symmetric distribution of matrices $\mathbf{A}$ and $\mathbf{B}$ in this work. A more general problem formulation where one can distribute different number of entries from $\mathbf{A}$ and $\mathbf{B}$ to each worker is an open problem.}. The computational complexities of the master and fusion nodes, in terms of the matrix parameter $N$, is required to be negligible in a scaling sense than the computational complexity at any worker node\footnote{If the master node or the fusion node is allowed to have higher computational complexity, the workers can simply store $\mathbf{A},\mathbf{B}$ using Maximum Distance Separable (MDS) codes to get a recovery threshold of $m$; the fusion node simply recovers $\mathbf{A},\mathbf{B}$ and then multiplies them, essentially performing the whole operation.}. The goal is to perform this matrix-matrix multiplication utilizing faulty or delay-prone workers with minimum recovery threshold.

\section{MatDot Codes}\label{sec:main}
In this section we will describe the distributed matrix-matrix product strategy using MatDot codes, and then examine computation and communication costs of the proposed strategy. Before proceeding further into the detailed construction and analyses of MatDot codes, we will first give some motivating examples which contrast MatDot codes with existing techniques.

\subsection{Motivating examples and summary of previous results}
Consider the problem statement described in Section~\ref{sec:sysmod}. We describe three different strategies as possible solutions to the problem: (i) ABFT matrix multiplication \cite{Huang_TC_84} (also called \emph{product-coded matrices} in \cite{ProductCodes}), (ii) Polynomial codes \cite{polynomialcodes} and then (iii) our proposed construction,  MatDot codes, each progressively improving, i.e., reducing the recovery threshold. We will evaluate the straggler tolerance of a strategy by its recovery threshold, $k$. For all the examples, we consider the most simple case with $m=2$.  Let us begin by describing the first strategy, namely, ABFT matrix multiplication.

\begin{example}[ABFT codes \cite{Huang_TC_84} ($m =2$, $k =  2\sqrt{P}$)]

Consider two $N \times N$ matrices $\mathbf{A}$ and $\mathbf{B}$ that are split as follows: $$\mathbf{A}=\begin{bmatrix}\mathbf{A}_0 \\ \mathbf{A}_{1}\end{bmatrix}, \mathbf{B} = \begin{bmatrix}\mathbf{B}_0 & \mathbf{B}_{1}\end{bmatrix}$$
 where $\mathbf{A}_0, \mathbf{A}_1$ are submatrices (row-blocks) of $\mathbf{A}$ of dimension $N/2 \times N$ and   $\mathbf{B}_0, \mathbf{B}_1$ are submatrices (column-blocks) of $\mathbf{B}$ of dimension $N \times N/2$.
Using ABFT, it is possible to compute $\mathbf{A}\mathbf{B}$ over $P$ nodes such that,
$\mathbf{(i)}$  each node uses $N^2/2$ linear combination of the entries of $\mathbf{A}$ and $N^2/2$ linear combination of the entries of $\mathbf{B}$ and $\mathbf{(ii)}$ the overall computation is tolerant to $P-2\sqrt{P}$ stragglers in the worst case. Thus, any $P-(P-2\sqrt{P})= 2\sqrt{P}$ worker nodes suffice to recover $\mathbf{A}\mathbf{B}$.

ABFT codes use the following strategy: $P$ processors are arranged in a $\sqrt{P} \times \sqrt{P}$ grid. ABFT codes encode two  row-blocks of $\mathbf{A}$ and two column-blocks of $\mathbf{B}$ separately using two systematic $(\sqrt{P},2)$ MDS codes. Then, we distribute the $i$-the encoded column-block of $\mathbf{A}$ to all the worker nodes on the $i$-th row of the grid, and the $j$-th encoded row-blocks to all the worker nodes on the $j$-th column of the grid. Note that here the grid indexing is $i=1,2,\dots,\sqrt{P}$ and $j=1,2,\ldots,\sqrt{P}$. An example for $P=9$ is shown in Fig.~\ref{fig:ABFTmatmult}. The worst case arises when all but one worker node in the lower right $(\sqrt{P}-1) \times (\sqrt{P}-1)$ part of the grid fail. Thus, the worst case recovery threshold is $P-(\sqrt{P}-1)^2+1 = 2\sqrt{P}$. For the example given in Fig. \ref{fig:ABFTmatmult} where $P=9$, recovery threshold is $2\sqrt{P} = 6$. \hfill $\blacksquare$

%The general principle of ABFT is to encode the rows of $\mathbf{A}$ and the columns of $\mathbf{B}$ separately using systematic \textcolor{red}{$(\sqrt{P},m)$} MDS codes \sout{of dimension $m$} .

\newsavebox{\smlmat}% Box to store smallmatrix content
\savebox{\smlmat}{$\begin{bmatrix} \mathbf{A}_{0}\\ \mathbf{A}_{1} \end{bmatrix}$}

\begin{figure}[htb!]
\centering
\includegraphics[width=0.5\textwidth]{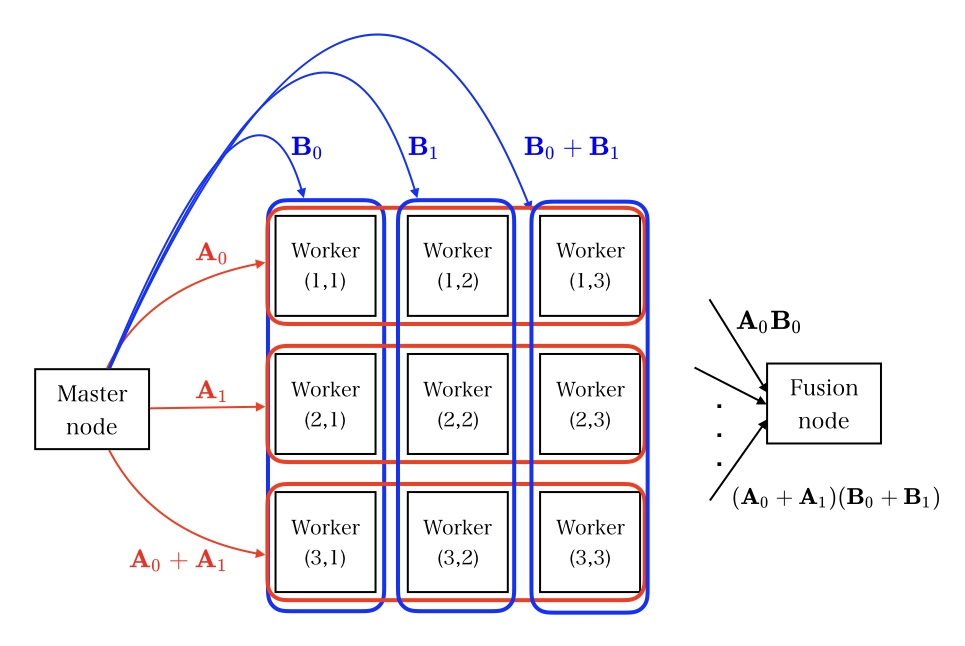}
\caption{ABFT matrix multiplication \cite{Huang_TC_84} for $P=9$ worker nodes with $m=2,$ where $\mathbf{A}=$\usebox{\smlmat}, $\mathbf{B}=[\mathbf{B}_{0}~~\mathbf{B}_{1}]$. The recovery threshold is $6.$ }
\label{fig:ABFTmatmult}
\end{figure}
\end{example}

In the previous example, the recovery threshold was a function of $P$ and thus it required more successful worker nodes as we use more processors. However, as we will show in the next example, Polynomial codes~\cite{polynomialcodes} provide a  superior recovery threshold that does not depend on $P$.

\begin{example}[Polynomial codes \cite{polynomialcodes} ($m = 2$, $k = 4$)]
Consider two $N \times N$ matrices $\mathbf{A}$ and $\mathbf{B}$ that are split as follows: $$\mathbf{A}=\begin{bmatrix}\mathbf{A}_0 \\ \mathbf{A}_{1}\end{bmatrix}, \mathbf{B} = \begin{bmatrix}\mathbf{B}_0 & \mathbf{B}_{1}\end{bmatrix}.$$
Polynomial codes computes $\mathbf{A}\mathbf{B}$ over $P$ nodes such that,
$\mathbf{(i)}$  each node uses $N^2/2$ linear combination of the entries of $\mathbf{A}$ and $N^2/2$ linear combination of the entries of $\mathbf{B}$ and $\mathbf{(ii)}$ the overall computation is tolerant to $P-4$ stragglers, i.e., any
$4$ nodes suffice to recover $\mathbf{A}\mathbf{B}$,
Polynomial codes use the following strategy: Node $i$ computes $(\mathbf{A}_{0}+\mathbf{A}_{1} i)(\mathbf{B}_0+\mathbf{B}_{1}i^{2}), i=1,2, \ldots P,$ so that from any $4$ of the $P$ nodes, the polynomial $p(x) = (\mathbf{A}_{0}\mathbf{B}_0+\mathbf{A}_{1}\mathbf{B}_0 x + \mathbf{A}_0\mathbf{B}_1 x^{2}+ \mathbf{A}_0\mathbf{B}_1 x^{3})$ can be interpolated. Having interpolated the polynomial, $\mathbf{A}\mathbf{B}$ as
$\begin{bmatrix}\mathbf{A}_{0}\mathbf{B}_{0} & \mathbf{A}_{0}\mathbf{B}_{1} \\ \mathbf{A}_{1}\mathbf{B}_{0} & \mathbf{A}_{1}\mathbf{B}_{1}\end{bmatrix}$ can be obtained from the coefficients (matrices) of the polynomial. $\blacksquare$

\newsavebox{\smlmata}% Box to store %smallmatrix content
\savebox{\smlmata}{$\begin{bmatrix}\tilde{\mathbf{A}}_{1}\\\tilde{\mathbf{A}}_{2}\end{bmatrix}$}
\begin{figure}[htb!]\label{fig:polycodes}
\centering
\includegraphics[width=0.5\textwidth]{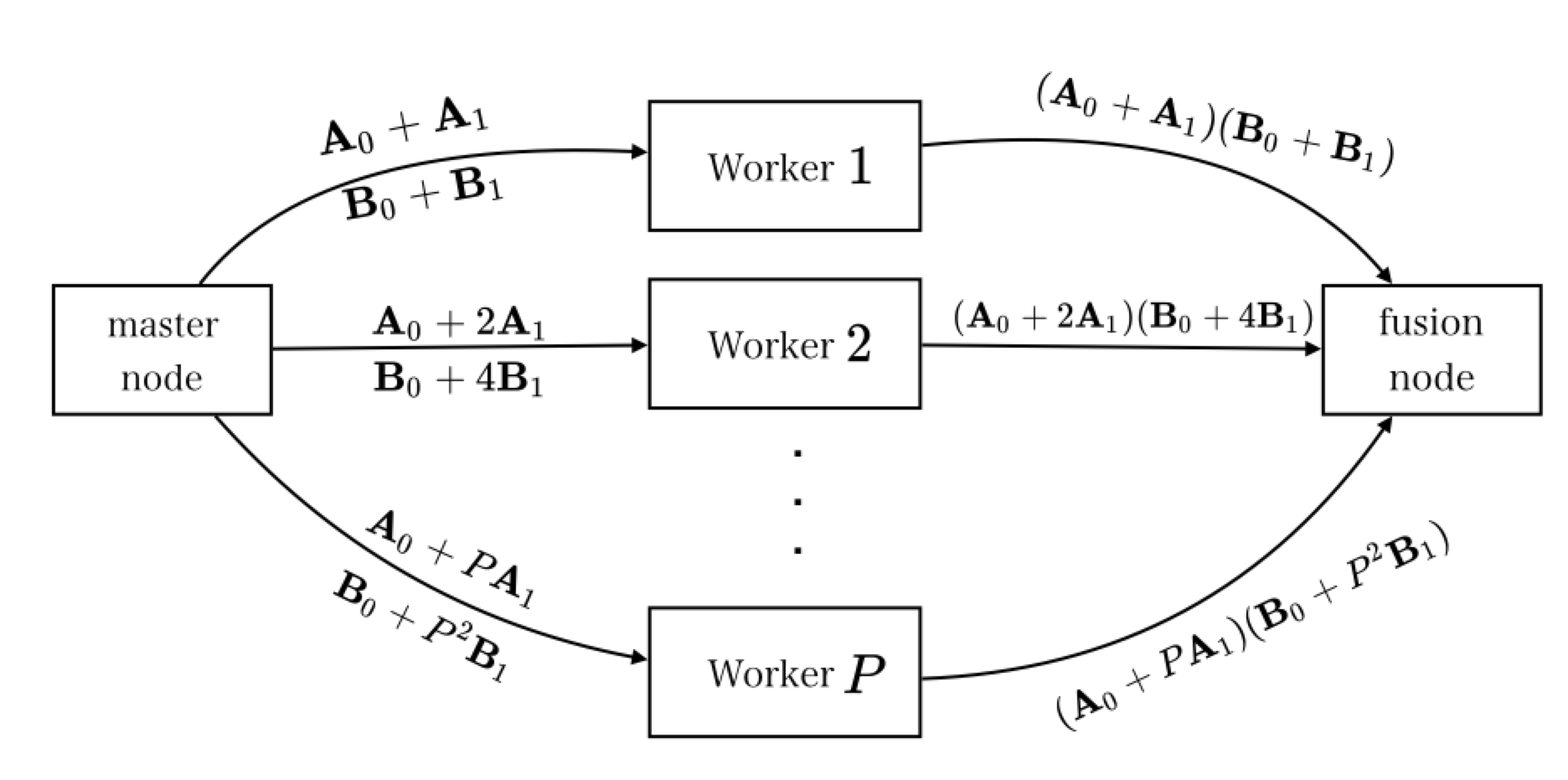}
\caption{Polynomial Codes \cite{polynomialcodes} with $m=2$. The recovery threshold is $4.$}
\end{figure}
\end{example}

Our novel \emph{MatDot} construction achieves a smaller recovery threshold as compared with Polynomial codes. Unlike ABFT and polynomial codes, MatDot divides matrix $\mathbf{A}$ vertically into column-blocks and matrix $\mathbf{B}$ horizontally into row-blocks.

\begin{example}\label{ex:matdot}[MatDot codes ($m=2$, $k=3$)]

MatDot codes compute $\mathbf{A}\mathbf{B}$ over $P$ nodes such that,
$\mathbf{(i)}$  each node uses $N^2/2$ linear combination of the entries of $\mathbf{A}$ and $N^2/2$ linear combination of the entries of $\mathbf{B}$ and $\mathbf{(ii)}$ the overall computation is tolerant to $P-3$ stragglers, i.e., $3$ nodes suffice to recover $\mathbf{A}\mathbf{B}$.
The proposed MatDot codes use the following strategy:
Matrix $\mathbf{A}$ is split vertically and $\mathbf{B}$ is split horizontally as follows:
\begin{equation}
   \mathbf{A} = \left[\mathbf{A}_0 \hspace{3mm} \mathbf{A}_1 \right],\;\;\; \mathbf{B}=\left[\begin{array}c \mathbf{B}_0\\\mathbf{B}_1\end{array}\right],
\end{equation}
where $\mathbf{A}_0, \mathbf{A}_1$ are submatrices (or column-blocks) of $\mathbf{A}$ of dimension $N \times N / 2$ and  $\mathbf{B}_0, \mathbf{B}_1$ are submatrices (or row-blocks) of $\mathbf{B}$ of dimension $N /2 \times N $.

Let $p_\mathbf{A}(x)= \mathbf{A}_0 + \mathbf{A}_1 x$ and $p_\mathbf{B}(x)= \mathbf{B}_0 x + \mathbf{B}_1$.  Let $x_1, x_2, \cdots, x_{P}$ be distinct real numbers, the master node sends $p_\mathbf{A}(x_{r})$ and $p_\mathbf{B}(x_{r})$ to the $r$-th worker node where the $r$-th worker node performs the multiplication $p_\mathbf{A}(x_{r})p_\mathbf{B}(x_{r})$ and sends the output to the fusion node. The exact computations at each worker node are depicted in Fig. \ref{fig:matdot}. We can observe that the fusion node can obtain the product $\mathbf{A}\mathbf{B}$ using the output of any three successful workers as follows: Let the worker nodes $1,2,$ and $3$ be the first three successful worker nodes, then the fusion node obtains the following three matrices:
\begin{align}
p_\mathbf{A}(x_{1})p_\mathbf{B}(x_{1})&= \mathbf{A}_0\mathbf{B}_1+(\mathbf{A}_0\mathbf{B}_0+\mathbf{A}_1\mathbf{B}_1)x_1+\mathbf{A}_1\mathbf{B}_0 x^2_1, \notag\\
p_\mathbf{A}(x_{2})p_\mathbf{B}(x_{2})&=\mathbf{A}_0\mathbf{B}_1+(\mathbf{A}_0\mathbf{B}_0+\mathbf{A}_1\mathbf{B}_1)x_2+\mathbf{A}_1\mathbf{B}_0 x^2_2, \notag\\
p_\mathbf{A}(x_{2})p_\mathbf{B}(x_{3})&=\mathbf{A}_0\mathbf{B}_1+(\mathbf{A}_0\mathbf{B}_0+\mathbf{A}_1\mathbf{B}_1)x_3+\mathbf{A}_1\mathbf{B}_0 x^2_3.\notag
\end{align}
Since these three matrices can be seen as three evaluations of the matrix polynomial $p_\mathbf{A}(x)p_\mathbf{B}(x)$ of degree $2$ at three distinct evaluation points $x_1, x_2, x_3$, the fusion node can obtain the coefficients of $x$ in $p_\mathbf{A}(x)p_\mathbf{B}(x)$ using polynomial interpolation. This includes the coefficient of $x$, which is $\mathbf{A}_0\mathbf{B}_0+\mathbf{A}_1\mathbf{B}_1 = \mathbf{A}\mathbf{B}$. Therefore, the fusion node can recover the matrix product $\mathbf{AB}$. \hfill $\blacksquare$
\end{example}
\begin{figure}[t]
\centering
\includegraphics[width=0.48\textwidth]{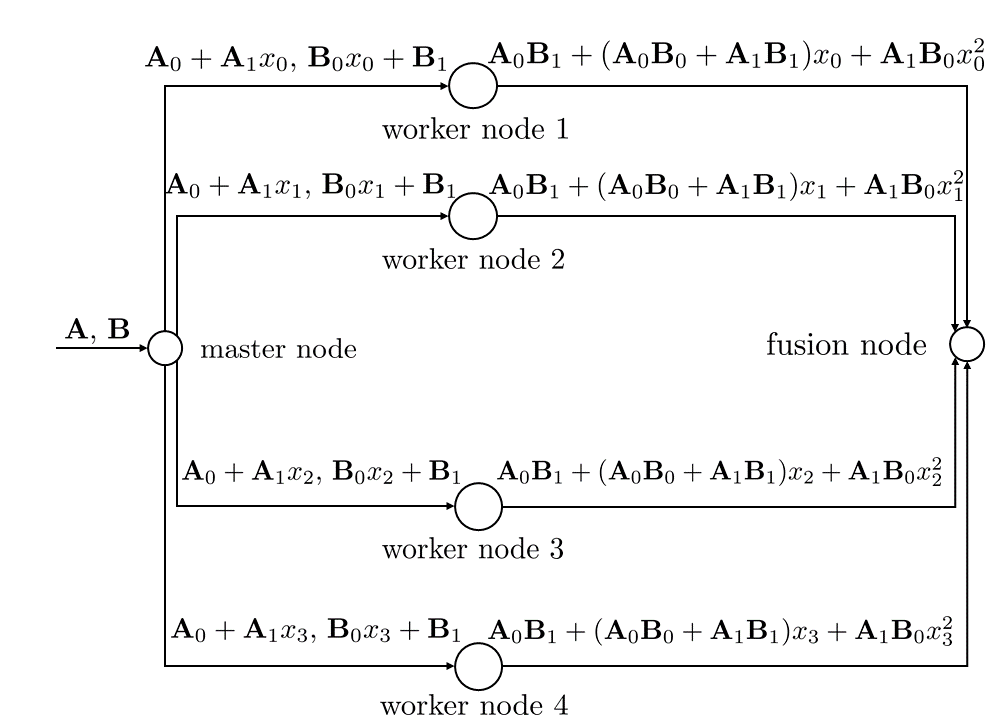}
\caption{An illustration of the computational system  with four worker nodes and applying MatDot codes with $m=2$. The recovery threshold is $3$.}
\label{fig:matdot}
\end{figure}

In the example, we have seen that for $m =2$, the recovery threshold of MatDot codes is $k = 3$ which is lower than Polynomial codes as well as ABFT matrix multiplication. The following theorem shows that for any integer $m$, the recovery threshold of MatDot codes is $k = 2m-1$.
\begin{theorem}\label{thm:matdotA}
For the matrix multiplication problem specified in Section~\ref{sec:problem} computed on the system defined in Definition~\ref{def:compsys},  a recovery threshold of $2m-1$ is achievable where  $m \geq 2$ is any positive integer that divides $N$.
\end{theorem}

Before we prove Theorem \ref{thm:matdotA}, we first describe the construction of MatDot codes.

%The following describes the general construction of MatDot Codes.
\begin{construction}\label{con:matdot}[MatDot Codes]

\textbf{Splitting of input matrices}:
 The matrix $\mathbf{A}$  is split vertically into $m$ equal column-blocks (of $N^2/m$ {symbols} each) and $\mathbf{B}$ is split horizontally into $m$ equal row blocks (of $N^2/m$ {symbols} each) as follows:
\begin{equation}\label{eq:spltA}
\mathbf{A} = \left[\mathbf{A}_0 \ \mathbf{A}_1 \ \ldots \ \mathbf{A}_{m-1}\right],\;\;\; \mathbf{B}=\left[\begin{array}c \mathbf{B}_0\\\mathbf{B}_1\\\vdots \\\mathbf{B}_{m-1}\end{array}\right],
\end{equation}
where, for $i \in \{0,\ldots,m-1\}$, and $\mathbf{A}_{i},\mathbf{B}_{i}$ are $N \times N/m$ and $N/m \times N$ dimensional submatrices, respectively.

\textbf{Master node (encoding)}: Let $x_1, x_2, \ldots, x_P$ be  distinct elements in $\mathbb{F}$. Let $p_\mathbf{A}(x)=\sum_{i=0}^{m-1} \mathbf{A}_i x^{i}$ and $p_\mathbf{B}(x)=\sum_{j=0}^{m-1} \mathbf{B}_j x^{m-1-j}.$ The master node sends to the $r$-th worker evaluations of $p_\mathbf{A}(x),p_\mathbf{B}(x)$ at $x=x_r$, that is, it sends $p_\mathbf{A}(x_r),p_\mathbf{B}(x_r)$ to the $r$-th worker.

\textbf{Worker nodes}:  For $r \in \{1,2,\ldots, P\}$, the $r$-th worker node computes the matrix product $p_{\mathbf{C}}(x_r)=p_\mathbf{A}(x_r) p_\mathbf{B}(x_r)$ and sends it to the fusion node on successful completion.

\textbf{Fusion node (decoding)}: The fusion node  uses outputs of any $2m-1$ successful workers to  compute the coefficient of $x^{m-1}$ in the product $p_{\mathbf{C}}(x)=p_\mathbf{A}(x) p_\mathbf{B}(x)$ (the feasibility of this step will be shown later in the proof of Theorem \ref{thm:matdotA}). If the number of successful workers is smaller than $2m-1$, the fusion node declares a failure.
\end{construction}

Notice that in MatDot codes, we have \begin{equation}\label{eq:dot}
\mathbf{AB}= \sum_{i=0}^{m-1} \mathbf{A}_i\mathbf{B}_i,
\end{equation}
where $\mathbf{A}_i$ and $\mathbf{B}_i$ are as defined in (\ref{eq:spltA}).
The simple observation of (\ref{eq:dot}) leads to a different way of computing the matrix product as compared with Polynomial codes based computation. In particular, to compute the product, we only require, for each $i$, the product of $\mathbf{A}_i$ and $\mathbf{B}_i.$ We do not require products of the form $\mathbf{A}_i\mathbf{B}_j$ for $i \neq j$ unlike Polynomial codes, where, after splitting the matrices $\mathbf{A},\mathbf{B}$ in to $m$ parts, \emph{all} $m^2$ cross-products are required to evaluate the overall matrix product. This leads to a significantly smaller recovery threshold for our construction.
\begin{proof}[Proof of Theorem \ref{thm:matdotA}]
To prove the theorem, it suffices to show that in the MatDot code construction described above, the fusion node is able to construct $\mathbf{C}$ from any $2m-1$ worker nodes. Observe that the coefficient of $x^{m-1}$ in:
\begin{equation}
p_{\mathbf{C}}(x) = p_\mathbf{A}(x) p_\mathbf{B}(x) = \left(\sum_{i=0}^{m-1} \mathbf{A}_i x^{i}\right)\left(\sum_{j=0}^{m-1} \mathbf{B}_j x^{m-1-j}\right)
\end{equation}
%\begin{remark}
is $\mathbf{AB}= \sum_{i=0}^{m-1} \mathbf{A}_i\mathbf{B}_i$ (from~\eqref{eq:dot}), which is the desired matrix-matrix product. Thus it is sufficient to compute this coefficient at the fusion node as the computation output for successful computation.  Now, because the polynomial $p_{\mathbf{C}}(x)$ has degree $2m-2$, evaluation of the polynomial at any $2m-1$ distinct points is sufficient to compute all of the coefficients of powers of $x$ in $p_\mathbf{A}(x)p_\mathbf{B}(x)$ using polynomial interpolation. This includes $\mathbf{AB}= \sum_{i=0}^{m-1} \mathbf{A}_i\mathbf{B}_i$, the coefficient of $x^{m-1}$. The next section has a complexity analysis that shows that the master and fusion nodes  have a lower computational complexity as compared to the worker nodes.
\end{proof}

\subsection{Complexity analyses of MatDot codes}\label{sec:complexity}
\textbf{Encoding/decoding complexity}:  Decoding requires interpolating a $2m-2$ degree polynomial for $N^2$ elements. Using polynomial interpolation algorithms of complexity $\mathcal{O}(k \log^2 k)$~\cite{kung1973fast},  where $k=2m-1$, the decoding complexity per matrix element is $\mathcal{O}(m \log^2 m )$. Thus, for $N^2$ elements, the decoding complexity is  $\mathcal{O}(N^2 m \log^2 m).$

Encoding for each worker requires performing two additions, each adding $m$ scaled matrices of size $N^2/m$, for an overall encoding complexity for \textit{each worker} of $\mathcal{O}(m N^2/m)=\mathcal{O}(N^2)$. Thus, the overall computational complexity of encoding for $P$ workers is $\mathcal{O}(N^2P)$.

\textbf{Each worker's computational cost}: Each worker multiplies two matrices of  dimensions $N\times N/m$ and $N/m\times N$, requiring $N^3/m$ operations (using straightforward multiplication algorithms\footnote{More sophisticated algorithms~\cite{strassen} also require super-quadratic complexity in $N$, so the conclusion will hold if those algorithms are used at workers as well.}). Hence, the computational complexity for each worker is $\mathcal{O}(N^3/m)$. Thus, as long as $P\ll N$ (and hence $m\ll N$), encoding and decoding complexity is  smaller than per-worker computational complexity.

\textbf{Communication cost}:\\
The master node communicates $\mathcal{O}(PN^2/m)$ symbols, and the fusion node receives $\mathcal{O}(mN^2)$ symbols from the successful worker nodes. While the master node communication is identical to that in Polynomial codes, the fusion node there only receives $\mathcal{O}(m^2N^2/m^2)=\mathcal{O}(N^2)$ symbols.

%\end{remark}

%\subsection{Recovery threshold of MatDot codes}

\subsection{Why does MatDot exceed the fundamental limits in~\cite{polynomialcodes}}
The fundamental limit in~\cite{polynomialcodes} concludes that the recovery threshold is $\Omega(m^2)$, whereas our recovery threshold is lower: $2m-1$. To understand why this is possible, one needs to carefully examine the derivation of the fundamental limit in~\cite{polynomialcodes}, which uses a cut-set argument to count the number of bits/symbols required for computing the product $\mathbf{AB}$. In doing so, the authors make the assumption that the number of symbols communicated by each worker to the fusion node is  $N^2/m^2$, which is a fallout of a horizontal division of matrix $\mathbf{A}$, and a vertical division of matrix $\mathbf{B}$ (the opposite of the division used here).

The bound  does not apply to our construction because each worker  now communicates $N^2$ symbols to the fusion node. Note that while the amount of information in each worker's transmissions is less, \textit{i.e.},  $O(N^2/m)$ (because the $N \times N$ matrices communicated by the workers can have rank $N^{2}/m$), this is still significantly larger than $N^2/m^2$ assumption made in the fundamental limits in~\cite{polynomialcodes}.

From a communication viewpoint, MatDot requires communicating a total of $(2m-1)N^2$ symbols, which is larger than the $N^2$ symbols in the product $\mathbf{AB}$. This is suggestive of a trade-off between minimal number of workers and minimal (sum-rate) communication from non-straggling workers. Section \ref{sec:polydot} describes a unified view of MatDot and Polynomial codes, which describes the trade-off between worker-fusion communication cost and recovery threshold achieved by our construction.

In practice, whether this increased worker-fusion node communication cost using MatDot codes is worth paying for will depend on the computational fabric and system implementation choices. Even in systems where communication costs may be significant, it is possible that more communication from fewer successful workers is less expensive than requiring more successful workers as required in Polynomial codes. Also note that if $P=\Omega(m^2)$ (e.g. when the system is highly fault prone or the deadline~\cite{SanghamitraISIT2017} is very short), communication complexity at the master node will dominate, and hence MatDot codes may not impose a substantial computing overhead.

\section{Systematic Code Constructions}\label{sec:syscod}

In this section, we provide a systematic code construction for MatDot codes. As the notion of systematic codes in the context of matrix multiplication problem is ambiguous, we will first define systematic codes in our context.
\begin{definition} [Systematic code in distributed matrix-matrix multiplication]
For the problem stated in Section \ref{sec:problem} computed on the system defined in Definition~\ref{def:compsys} such that the matrices $\mathbf{A}$ and $\mathbf{B}$ are split as in (\ref{eq:spltA}), a code is called \emph{systematic} if the output of the $r$-th worker node is the product $\mathbf{A}_{r-1}\mathbf{B}_{r-1}$, for all $r \in \{1, \cdots, m\}$. We refer to the first $m$ worker nodes, that output $\mathbf{A}_{r-1}\mathbf{B}_{r-1}$ for $r \in \{1, \cdots, m\}$, as \emph{systematic worker nodes}.
\end{definition}

Note that the final output $\mathbf{A} \mathbf{B}$ can be obtained by summing up the outputs from the $m$ systematic worker nodes:  $$\mathbf{AB}= \sum_{r=1}^{m} \mathbf{A}_{r-1}\mathbf{B}_{r-1}.$$

The presented systematic code, named \emph{``systematic MatDot code''}, is advantageous over MatDot codes in two aspects. Firstly,  even though both MatDot and systematic MatDot codes have the same recovery threshold, systematic MatDot codes can recover the output as soon as the $m$ systematic worker nodes successfully finish unlike MatDot codes which \emph{always} require $2m-1$ workers to successfully finish to recover the final result. Furthermore, when the $m$ systematic worker nodes successfully finish first, the decoding complexity using systematic MatDot codes is $O(m N^2)$, which is less than the decoding complexity of MatDot codes, \textit{i.e.}, $O(m N^2\log^2 m)$.
Another advantage for systematic MatDot codes over MatDot codes is that the systematic MatDot approach may be useful for \emph{backward-compatibility with current practice}. What this means is that, for systems that are already established and operating with no straggler tolerance, but do an $m$-way parallelization, it is easier to apply the systematic approach as the infrastructure could be appended to additional worker nodes without modifying what the first $m$ nodes are doing.

The following theorem shows that there exists a systematic MatDot code construction that achieves the same recovery threshold as MatDot codes.

\begin{theorem}\label{thm:syscodes}
For the matrix-matrix multiplication problem specified in Section~\ref{sec:problem} computed on the system defined in Definition~\ref{def:compsys}, there exists a systematic code, where the product $\mathbf{A}\mathbf{B}$ is the summation of the output of the first $m$ worker nodes,  that solves this problem with a recovery threshold {of} $2m-1$, where  $m \geq 2$ is any positive integer that divides $N$.
\end{theorem}

Before we describe the construction of systematic MatDot codes, that will be used to prove Theorem \ref{thm:syscodes}, we first present a simple example to illustrate the idea of systematic MatDot codes.

\begin{example}\label{ex:syscod}[Systematic MatDot code, $m=2,k=3$]

{Matrix $\mathbf{A}$ is split vertically into two submatrices (column-blocks) $\mathbf{A}_0$ and $\mathbf{A}_1$, each of dimension $N \times \frac{N}{2}$ and matrix $\mathbf{B}$ is split horizontally into two submatrices (row-blocks) $\mathbf{B}_0$ and  $\mathbf{B}_1$, each of dimension $\frac{N}{2} \times N $  as follows: }
\begin{equation}
   \mathbf{A} = \left[ \ \mathbf{A}_0 \  \ \mathbf{A}_1 \ \right],\;\;\; \mathbf{B}=\left[\begin{array}c \mathbf{B}_0\\\mathbf{B}_1\end{array}\right].
\end{equation}

Now, we define the encoding functions $p_\mathbf{A}(x)$ and $p_\mathbf{B}(x)$ as $p_\mathbf{A}(x)= \mathbf{A}_0  \frac{x-x_2}{x_1 - x_2}+ \mathbf{A}_1  \frac{x-x_1}{x_2 - x_1}$ and $p_\mathbf{B}(x)= \mathbf{B}_0  \frac{x-x_2}{x_1 - x_2} + \mathbf{B}_1 \frac{x-x_1}{x_2 - x_1}$, for distinct $x_1, x_2 \in \mathbb{F}$.  Let $x_3,  \cdots, x_{P}$ be  elements of $\mathbb{F}$ such that $x_1, x_2, x_3, \cdots, x_{P}$ are distinct, the master node sends $p_\mathbf{A}(x_{r})$ and $p_\mathbf{B}(x_{r})$ to the $r$-th worker node, $r \in \{1, \cdots, P\}$, where the $r$-th worker node performs the multiplication $p_\mathbf{A}(x_{r}) p_\mathbf{B}(x_{r})$ and sends the output to the fusion node. The exact computations at each worker node are depicted in Fig. \ref{fig:syscod}.

We can observe that the outputs of the worker nodes $1,2$ are $\mathbf{A}_0\mathbf{B}_0, \mathbf{A}_1\mathbf{B}_1$, respectively, and hence this code is systematic. Let us consider a scenario where the systematic worker nodes, \textit{i.e.}, worker nodes $1$ and $2$ complete their computations first. In this scenario, the fusion node does not require a decoding step and can obtain the product $\mathbf{A}\mathbf{B}$ by simply performing the summation of the two outputs it has received: $\mathbf{A}_0\mathbf{B}_0+ \mathbf{A}_1\mathbf{B}_1$. Now let us consider a different scenario where worker nodes $1,3,4$ are the first three successful workers. Then, the fusion node receives three matrices,
$p_\mathbf{A}(x_{1})p_\mathbf{B}(x_{1}), p_\mathbf{A}(x_{3})p_\mathbf{B}(x_{3}),$ and $p_\mathbf{A}(x_{4})p_\mathbf{B}(x_{4})$. Since these three matrices can be seen as three evaluations of the polynomial $p_\mathbf{A}(x)p_\mathbf{B}(x)$ of degree $2$ at three distinct evaluation points $x_1, x_3, x_4$, the coefficients of the polynomial $p_\mathbf{A}(x)p_\mathbf{B}(x)$ can be obtained using polynomial interpolation. Finally, to obtain the product $\mathbf{A}\mathbf{B}$, we evaluate   $p_\mathbf{A}(x)p_\mathbf{B}(x)$ at $x = x_1, x_2$ and sum them up:
\begin{align*}
p_\mathbf{A}(x_{1})p_\mathbf{B}(x_{1})+p_\mathbf{A}(x_{2})p_\mathbf{B}(x_{2})
=\mathbf{A}_0\mathbf{B}_0+\mathbf{A}_1 \mathbf{B}_1
= \mathbf{A}\mathbf{B}.
\end{align*}\hfill $\blacksquare $
\end{example}

\begin{figure}[t]
\centering
\includegraphics[width=0.48\textwidth]{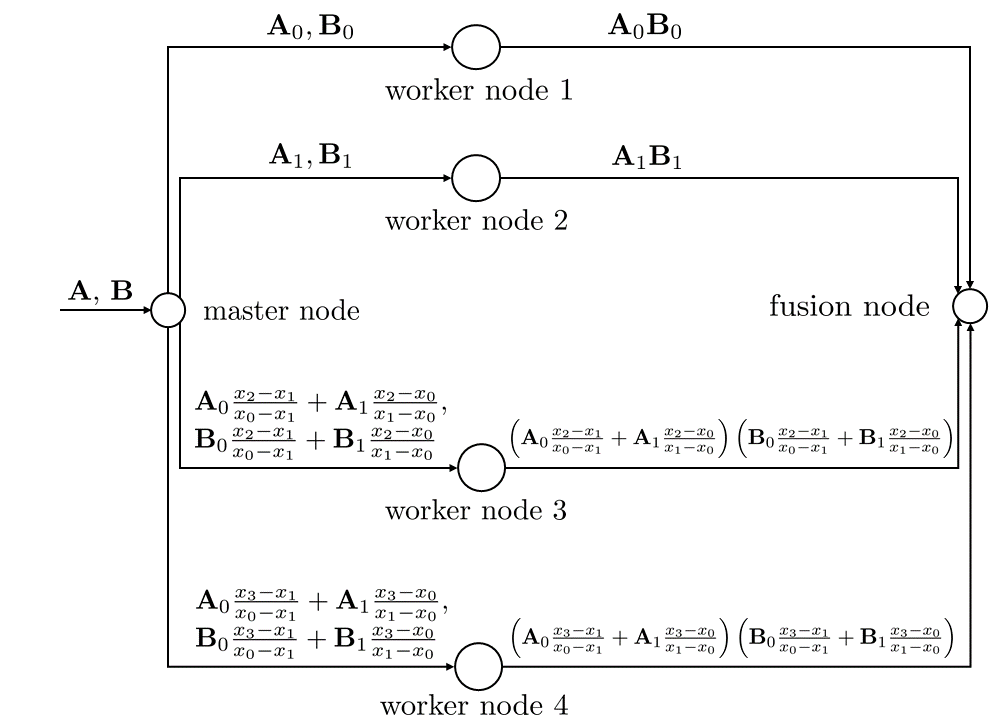}
\caption{An illustration of the computational system  with four worker nodes and applying systematic MatDot codes with $m=2$. The recovery threshold is $3$.}
\label{fig:syscod}
\end{figure}

The following describes the general construction of the systematic MatDot codes for matrix-matrix multiplication.

\begin{construction}\label{con:syscod}[Systematic MatDot codes]
\text{}

\textbf{Splitting of input matrices}:
 $\mathbf{A}$ and $\mathbf{B}$ are split as in (\ref{eq:spltA}). %We can observe that $\mathbf{AB}= \sum_{i=0}^{m-1} \mathbf{A}_i\mathbf{B}_i$.
%Now, we describe the computation and communication of MatDot Codes.

\textbf{Master node (encoding)}: Let $x_1, x_2, \ldots, x_{P}$ be arbitrary distinct elements of $\mathbb{F}$. Let $p_\mathbf{A}(x)=\sum_{i=1}^{m} \mathbf{A}_{i-1} L_{i}(x)$ and $p_\mathbf{B}(x)=\sum_{i=1}^{m} \mathbf{B}_{i-1} L_{i}(x)$ where $L_i(x)$ is defined as follows for $i \in \{ 1,\cdots, m \}$:
\begin{equation}\label{eq:lpol}
    L_i(x)=\displaystyle\prod\limits_{\substack{  j\in \{1, \cdots, m\} \setminus \{i\}}}\frac{ x-x_j}{x_i-x_j}.
\end{equation}

The master node sends to the $r$-th worker the evaluations of $p_\mathbf{A}(x),p_\mathbf{B}(x)$ at $x=x_{r}$, \textit{i.e.}, it sends $p_\mathbf{A}(x_{r}),p_\mathbf{B}(x_{r})$ to the $r$-th worker for $r \in \{1,2,\ldots, P\}$.

\textbf{Worker nodes}:  For $r \in \{1,2,\ldots, P\}$, the $r$-th worker node computes the matrix product $\mathbf{C}(x_{r})=p_\mathbf{A}(x_{r}) p_\mathbf{B}(x_{r})$ and sends it to the fusion node on successful completion.

\textbf{Fusion node (decoding)}:  For any $k$ such that $m \leq k \leq 2m-1$, whenever the outputs of the first $k$ successful workers contain the outputs  of the systematic worker nodes $1, \cdots, m$, \textit{i.e.}, $\{\mathbf{C}(x_{r})\}_{r \in \{1, \cdots, m\}}$ is contained in the set of the first $k$ outputs received by the fusion node, the fusion node performs the summation $\sum_{r=1}^{m}\mathbf{C}(x_r)$. Otherwise, if $\{\mathbf{C}(x_{r})\}_{r \in \{1, \cdots, m\}}$ is not contained in the set of the first $2m-1$ evaluations received by the fusion node, the  fusion node performs the following steps: (i) interpolates the polynomial $\mathbf{C}(x)=p_{\mathbf{A}}(x)p_{\mathbf{B}}(x)$ (the feasibility of this step will be shown later in the proof of Theorem \ref{thm:syscodes}), (ii)  evaluates  $\mathbf{C}(x)$ at $x_1, \cdots, x_{m}$, (iii)   performs the summation $\sum_{r=1}^{m} \mathbf{C}(x_r)$.

If the number of successful worker nodes is smaller than $2m-1$ and the first $m$ worker nodes are not included in the successful worker nodes, the fusion node declares a failure.
\end{construction}

The following lemma proves that the construction given here is systematic.

\begin{lemma}\label{rmk:syscod1}
For Construction \ref{con:syscod}, the output of the $r$-th worker node, for $r \in \{1, \cdots, m\}$, is the product $\mathbf{A}_{r-1}\mathbf{B}_{r-1}$. That is, Construction \ref{con:syscod} is a systematic code for matrix-matrix multiplication.
\end{lemma}
\begin{proof}[Proof of Lemma \ref{rmk:syscod1}]
The lemma follows from noting that the output of the $r$-th worker, for $r \in \{1, \cdots, m\}$, can be written as
\begin{align}
p_{\mathbf{C}}(x_{r})&=p_\mathbf{A}(x_{r}) p_\mathbf{B}(x_{r})\notag\\
&=\sum_{i=1}^{m} \mathbf{A}_{i-1} L_{i}(x_{r})\sum_{i=1}^{m} \mathbf{B}_{i-1} L_{i}(x_{r})\notag\\
&\stackrel{}{=}\mathbf{A}_{r-1}\mathbf{B}_{r-1},
\end{align}
where the last equality follows from the property of $L_i(x)$:
\begin{equation}
    L_{i}(x_j)=\begin{cases}
      1 & \text{if } j=i, \\
      0 & \text{if } j \in \{1, \cdots, m \} \setminus \{i\},
\end{cases}
\end{equation}
for $i \in \{ 1, \cdots, m\}$.
\end{proof}

Now, we  proceed with the proof of Theorem \ref{thm:syscodes}.
\begin{proof}[Proof of Theorem \ref{thm:syscodes}]
Since Construction \ref{con:syscod} is a systematic code for matrix-matrix multiplication (Lemma \ref{rmk:syscod1}), in order to prove the theorem, it suffices to show that  Construction \ref{con:syscod} is a valid construction with a recovery threshold $k = 2m-1$. From (\ref{eq:lpol}), observe that the polynomials $L_i(x)$, $i \in \{1, \cdots, m \}$, have degrees $m-1$ each. Therefore, each of $p_{\mathbf{A}}(x)=\sum_{i=1}^{m} \mathbf{A}_{i-1} L_{i}(x)$ and $p_{\mathbf{B}}(x)=\sum_{i=1}^{m} \mathbf{B}_{i-1} L_{i}(x)$ has a degree of $m-1$ as well. Consequently, $p_{\mathbf{C}}(x)=p_{\mathbf{A}}(x)p_{\mathbf{B}}(x)$ has a degree of $2m-2$.
 Now, because the polynomial $p_{\mathbf{C}}(x)$ has degree $2m-2$, evaluation of the polynomial at any $2m-1$ distinct points is sufficient to interpolate $\mathbf{C}(x)$ using polynomial interpolation algorithm.
 Now, since Construction \ref{con:syscod} is systematic (Lemma \ref{rmk:syscod1}), the product $\mathbf{A}\mathbf{B}$ is the summation of the outputs of the first m workers, \textit{i.e.}, $\mathbf{A}\mathbf{B}= \sum_{r=1}^{m} p_{\mathbf{C}}(x_{r})$. Therefore, after the fusion node interpolates $\mathbf{C}(x)$, evaluating $p_{\mathbf{C}}(x)$ at $x_1, \cdots, x_{m}$, and  performing the summation $\sum_{r=1}^{m} p_{\mathbf{C}}(x_r)$ yields the product $\mathbf{A}\mathbf{B}$.
\end{proof}
\subsection{Complexity analyses of the systematic codes}\label{sec:syscomplexity}
Apart from  the encoding/decoding complexity, the complexity analyses of sytstematic MatDot codes are the same as their MatDot codes counterparts. In the following, we investigate the encoding/decoding complexity of Construction \ref{con:syscod}.

\textbf{Encoding/decoding Complexity}:
Encoding for each worker first requires  performing evaluations of polynomials $L_i(x)$ for all $i\in \{1, \cdots,m \}$, with each evaluation requiring $\mathcal{O}(m)$ operations. This gives $\mathcal{O}(m^2)$ operations for all polynomial evaluations. Afterwards, two additions are performed, each adding $m$ scaled matrices of size $N^2/m$, with complexity $\mathcal{O}(m N^2/m)=\mathcal{O}(N^2)$. Therefore, the overall encoding complexity for \textit{each worker} is $\mathcal{O}(\max(N^2,m^2))=\mathcal{O}(N^2)$. Thus, the overall computational complexity of encoding for $P$ workers is $\mathcal{O}(N^2P)$.

For decoding, first, for the interpolation step, we interpolate a $2m-2$ degree polynomial for $N^2$ elements. Using polynomial interpolation algorithms of complexity $\mathcal{O}(k \log^2 k)$~\cite{kung1973fast},  where $k=2m-1$, the interpolation complexity per matrix element is $\mathcal{O}(m \log^2 m )$. Thus, for $N^2$ elements, the interpolation complexity is  $\mathcal{O}(N^2 m \log^2 m).$  For the evaluation of $p_{\mathbf{C}}(x)$ at $x_1, \cdots, x_{m}$, each evaluation involves
adding $2m-1$ scaled matrices of size $N^2$ with a complexity of $\mathcal{O}(mN^2)$. Hence, for all $m$ evaluations the complexity is $\mathcal{O}(m^2N^2)$. Finally, the complexity of the final  addition of  $m$ matrices of size $N^2$ is $\mathcal{O}(m N^2)$. Hence, the overall decoding complexity is $\mathcal{O}(m^2N^2)$.

\section{Unifying Matdot and Polynomial Codes: Trade-off between communication cost and recovery threshold}\label{sec:polydot}

In this section, we present a code construction, named \emph{PolyDot}, that provides a trade-off between communication costs and recovery thresholds. Polynomial codes \cite{polynomialcodes} have a higher recovery threshold of $m^2$, but have a lower communication cost of $\mathcal{O}(N^2/m^2)$ per worker node. On the other hand, MatDot codes have a lower recovery threshold of $2m-1$, but have a higher communication cost of $\mathcal{O}(N^2)$ per worker node. Here, our goal is to construct a code that bridges the gap between Polynomial codes and MatDot codes so that we can get intermediate communication costs and recovery thresholds, with Polynomial and MatDot codes as two special cases. For this goal, we propose PolyDot codes, which may be viewed as an interpolation of MatDot codes and Polynomial codes -- one extreme of the interpolation is MatDot codes and the other extreme is Polynomial codes.

We follow the same problem setup and system assumptions in \ref{sec:problem}.
The following theorem shows the recovery threhsold of PolyDot codes.

\begin{theorem}\label{thm:tdoff}%[Recovery threshold of PolyDot codes]
For the matrix multiplication problem specified in Section~\ref{sec:problem} computed on the system defined in Definition~\ref{def:compsys}, there exist codes with a recovery threshold of  $t^2(2s-1)$ and a communication cost from each worker node to the fusion node bounded by $\mathcal{O}(N^2/t^2)$ for any positive integers $s$, $t$ such that $st = m$ and both $s$ and $t$ divide $N$.
\end{theorem}

Before describing the PolyDot code construction and prove Theorem \ref{thm:tdoff}, we first introduce the following simple PolyDot code example with $m=4$ and $s = t= 2$.

\begin{example}\label{ex:polydot}[PolyDot codes ($m=4, s=2, k=12$)]

Matrix $\mathbf{A}$ is split into submatrices $\mathbf{A}_{0,0}, \mathbf{A}_{0,1},\mathbf{A}_{1,0},\mathbf{A}_{1,1}$  each of dimension $N /2 \times N / 2$. Similarly,  matrix $\mathbf{B}$ is split into submatrices $\mathbf{B}_{0,0}, \mathbf{B}_{0,1},\mathbf{B}_{1,0},\mathbf{B}_{1,1}$  each of dimension $N /2 \times N / 2$  as follows:
\begin{align}\label{eq:exinterAnB}
\mathbf{A} &= \left[\begin{array}{cc} \mathbf{A}_{0,0} &\mathbf{A}_{0,1}\\\mathbf{A}_{1,0}& \mathbf{A}_{1,1}\end{array}\right], \mathbf{B} = \left[\begin{array}{cc} \mathbf{B}_{0,0}  &\mathbf{B}_{0,1}\\\mathbf{B}_{1,0}& \mathbf{B}_{1,1}\end{array}\right].
\end{align}
Notice that, from (\ref{eq:exinterAnB}), the product $\mathbf{A}\mathbf{B}$ can be written as
\begin{align}\label{eq:exinterAB}
\mathbf{A} \mathbf{B} &= \left[\begin{array}{cc} \sum_{i=0}^1 \mathbf{A}_{0,i}\mathbf{B}_{i,0} &\sum_{i=0}^1 \mathbf{A}_{0,i}\mathbf{B}_{i,1}\\\sum_{i=0}^1 \mathbf{A}_{1,i}\mathbf{B}_{i,0}& \sum_{i=0}^1 \mathbf{A}_{1,i}\mathbf{B}_{i,1}\end{array}\right].
\end{align}

Now, we define the encoding functions $p_\mathbf{A}(x)$ and $p_\mathbf{B}(x)$ as
$$p_\mathbf{A}(x)= \mathbf{A}_{0,0}+\mathbf{A}_{1,0} x+\mathbf{A}_{0,1} x^2  + \mathbf{A}_{1,1} x^3,$$  $$p_\mathbf{B}(x)= \mathbf{B}_{0,0} x^2+\mathbf{B}_{1,0}  +\mathbf{B}_{0,1} x^8 +\mathbf{B}_{1,1} x^6.$$

Observe the following:
\begin{enumerate}[label=(\roman*)]
    \item the coefficient of $x^{2}$ in $p_\mathbf{A}(x)p_\mathbf{B}(x)$ is $\sum_{i=0}^1 \mathbf{A}_{0,i}\mathbf{B}_{i,0}$,
    \item the coefficient of $x^{8}$ in $p_\mathbf{A}(x)p_\mathbf{B}(x)$ is $\sum_{i=0}^1 \mathbf{A}_{0,i}\mathbf{B}_{i,1}$,
    \item the coefficient of $x^{3}$ in $p_\mathbf{A}(x)p_\mathbf{B}(x)$ is $\sum_{i=0}^1 \mathbf{A}_{1,i}\mathbf{B}_{i,0}$, and
    \item the coefficient of $x^{9}$ in $p_\mathbf{A}(x)p_\mathbf{B}(x)$ is $\sum_{i=0}^1 \mathbf{A}_{1,i}\mathbf{B}_{i,1}$.
\end{enumerate}

Let $x_1, \cdots, x_{P}$ be distinct elements of $\mathbb{F}$. The master node sends $p_\mathbf{A}(x_{r})$ and $p_\mathbf{B}(x_{r})$ to the $r$-th worker node, $r \in \{1, \cdots, P\}$, and the $r$-th worker node performs the multiplication $p_\mathbf{A}(x_{r})p_\mathbf{B}(x_{r})$ and sends the result to the fusion node.

%Observe that the coefficient of $x^{r+2+6c}$ in $p_\mathbf{A}(x)p_\mathbf{B}(x)$ is $\sum_{i=0}^1 \mathbf{A}_{r,i}\mathbf{B}_{i,c}$, for $r,c \in \{0,1\}$.

 Let worker nodes $1, \cdots, 12$ be the first $12$ worker nodes to send their computation outputs to the fusion node, then the fusion node obtains the matrices
$p_\mathbf{A}(x_{r})p_\mathbf{B}(x_{r})$ for all  $r \in \{1, \cdots, 12\}$. Since these $12$ matrices can be seen as twelve evaluations of the matrix polynomial $p_\mathbf{A}(x)p_\mathbf{B}(x)$ of degree $11$ at twelve distinct points $x_1, \cdots, x_{12}$, the coefficients of the matrix polynomial $p_\mathbf{A}(x) p_\mathbf{B}(x)$ can be obtained using polynomial interpolation. This includes the coefficients of $x^{i+2+6j}$ for all $i,j \in \{0,1\}$, i.e., $\sum_{k=0}^1 \mathbf{A}_{i,k}\mathbf{B}_{k,j}$ for all $i,j \in \{0,1\}$. Once the matrices $\sum_{k=0}^1 \mathbf{A}_{i,k} \mathbf{B}_{k,j}$ for all $i,j \in \{0,1\}$ are obtained, the product $\mathbf{A} \mathbf{B}$ is obtained by (\ref{eq:exinterAB}).\hfill $\blacksquare$
\end{example}

The recovery threshold for $m=4$ in Example~\ref{ex:polydot} is $k = 12$. This is larger than the recovery threshold of MatDot codes, which is $k = 2m-1 = 9$, and smaller then the recovery threshold of Polynomial codes, which is $k= m^2 = 16$. Hence, we can see that the recovery thresholds of PolyDot codes are somewhere between those of MatDot codes and Polynomial codes.

The following describes the general construction of PolyDot($m,s,t$) codes. Note that although two parameters $m$ and $s$ are sufficient to characterize a PolyDot code, we include $t$ in the parameters for better readability.

\begin{construction}\label{con:polydot}[PolyDot($m,s,t$) codes]

\textbf{Splitting of input matrices}: $\mathbf{A}$ and $\mathbf{B}$  are split both horizontally and vertically:
\begin{align}\label{eq:interAnB}
\mathbf{A} &= \left[\begin{array}{ccc} \mathbf{A}_{0,0} & \cdots &\mathbf{A}_{0,s-1}\\\vdots & \ddots & \vdots \\\mathbf{A}_{t-1,0}&\cdots& \mathbf{A}_{t-1,s-1}\end{array}\right],\notag\\\notag\\ \mathbf{B} &= \left[\begin{array}{ccc} \mathbf{B}_{0,0} & \cdots &\mathbf{B}_{0,t-1}\\\vdots & \ddots & \vdots \\\mathbf{B}_{s-1,0}&\cdots& \mathbf{B}_{s-1,t-1}\end{array}\right],
\end{align}
where, for $i = 0, \cdots, s-1, j = 0, \cdots, t-1$, $\mathbf{A}_{j,i}$'s are $N/t \times N/s$ submatrices of $\mathbf{A}$ and $\mathbf{B}_{i,j}$'s are $N/s \times N/t$ submatrices of $\mathbf{B}$. We choose $s$ and $t$ such that   both $s$ and $t$ divide $N$ and $st = m$.

\textbf{Master node (encoding)}: Define the encoding polynomials as:
\begin{align}
p_\mathbf{A}(x, y) &= \sum_{i=0}^{t-1}\sum_{j=0}^{s-1} \mathbf{A}_{i,j} x^{i} y^{j}, \nonumber \\
p_\mathbf{B}(y, z) &= \sum_{k=0}^{s-1}\sum_{l=0}^{t-1} \mathbf{B}_{k,l} y^{s-1-k} z^{l}.\label{eq:polydot_AB}
\end{align}
The master node sends to the $r$-th worker the evaluations of $p_\mathbf{A}(x, y), p_\mathbf{B}(y, z)$ at $x=x_r, y= x_r^t, z= x_r^{t(2s-1)}$ where $x_r$'s are all distinct for $r \in \{1,2,\ldots, P\}$. By this substitution, we are transforming the three-variable polynomial to a single-variable polynomial as follows\footnote{An alternate substitution can reduce the recovery threshold further as mentioned in subsequent works \cite{DNNPaperISIT, entangledpolycodes}. We will clarify this in Remark~\ref{rem:recovery_threshold}.}:
\begin{equation*}
p_{\mathbf{C}}(x, y, z) = p_{\mathbf{C}}(x) = \sum_{i,j,k,l} \mathbf{A}_{i,j} \mathbf{B}_{k,l} x^{i + t(s-1+j-k) + t(2s-1)l},
\end{equation*}
and evaluate the polynomial $\mathbf{C}(x)$ at $x_r$ for $r= 1, \cdots, P$. In  Lemma~\ref{lem:bijection} that this transformation is one-to-one.

\textbf{Worker nodes}:  For $r \in \{1,2,\ldots, P\}$, the $r$-th worker node computes the matrix product $p_{\mathbf{C}}(x_r, y_r, z_r)=p_\mathbf{A}(x_r, y_r) p_\mathbf{B}(y_r,z_r)$ and sends it to the fusion node on successful completion.

\textbf{Fusion node (decoding)}: The fusion node  uses outputs of the first $t^2(2s-1)$ successful workers to  compute the coefficient of $x^{i-1} y^{s-1} z^{l-1}$ in $\mathbf{C}(x,y,z)=p_\mathbf{A}(x, y) p_\mathbf{B}(y,z)$.  That is, it computes the coefficient of $x^{i-1+(s-1)t+ (2s-1)t(l-1)}$ of the transformed single-variable polynomial. The proof of Theorem \ref{thm:tdoff} shows that this is indeed possible. If the number of successful workers is smaller than $t^2(2s-1)$, the fusion node declares a failure.
\end{construction}

\begin{figure}[t]
\centering
\includegraphics[width=0.5\textwidth]{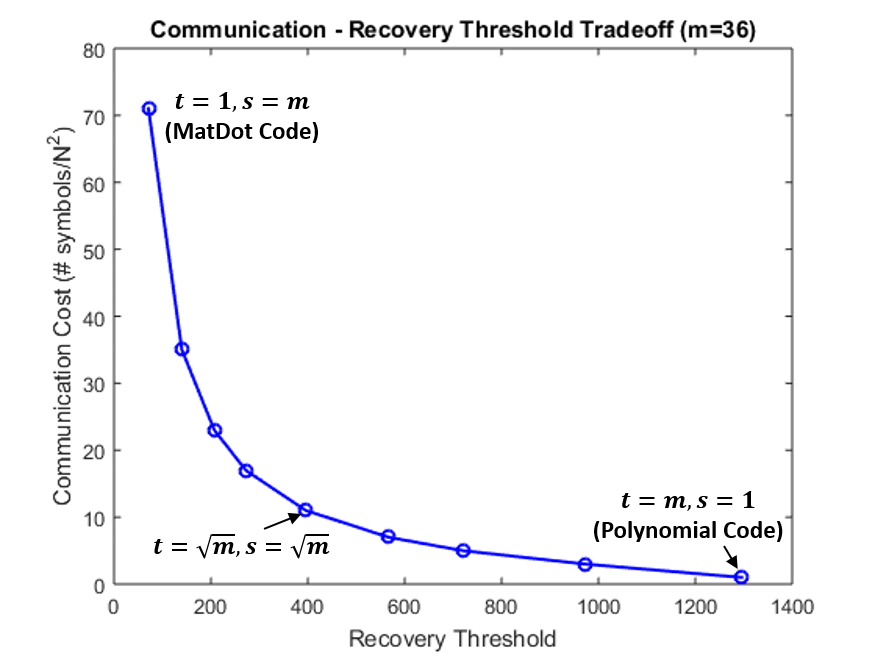}
\caption{An illustration of the trade-off between communication cost (from the workers to the fusion node) and recovery threshold of PolyDot codes by varying $s$ and $t$ for a fixed $m$ ($m = 36$). The minimum communication cost is $N^2$, corresponding to polynomial codes, that have the largest recovery threshold. It is important to note here that in the above, we are \textit{only including the communication cost from the workers to the fusion node}. The communication from the master node to the workers is not included, and it can dominate in situations when the workers are highly unreliable. }
\label{fig:tradeoff}
\end{figure}
Before we prove the theorem, let us discuss the utility of PolyDot codes. By choosing different $s$ and $t$, we can trade off communication cost and recovery threshold. For $s =m$ and $t =1$, PolyDot$(m, s=m, t=1)$ code is a MatDot code which has a low recovery threshold but high communication cost. At the other extreme, for $s =1$ and $t =m$, PolyDot$(m, s=1, t=m)$ code is a Polynomial code. Now let us consider a code with intermediate $s$ and $t$ values such as $s = \sqrt{m}$ and $t = \sqrt{m}$.
A PolyDot$(m, s=\sqrt{m}, t= \sqrt{m})$ code has a recovery threshold of $m(2\sqrt{m}-1) = \Theta(m^{1.5})$, and the total number of symbols to be communicated to the fusion node is $\Theta\left( (N/\sqrt{m})^2 \cdot m^{1.5} \right) = \Theta(\sqrt{m} N^2)$, which is smaller than $\Theta(m N^2)$ required by MatDot codes but larger than $\Theta(N^2)$ required by Polynomial codes. This trade-off is illustrated in Fig.~\ref{fig:tradeoff} for $m = 36$.

To prove Theorem \ref{thm:tdoff}, we need the following lemma.
% \begin{lemma} \label{lem:bijection}
% The following function
% \begin{align}
% \nonumber f: &\{0,\cdots, t-1\} \times \{0, \cdots, 2s-2\} \times \{0,\cdots, t-1\} \\& \to \{0,\cdots, t^2(2s-1)-1\} \nonumber\\
% & (i, \widetilde\jmath, l) \mapsto i + t\tilde{j} + t(2s-1)l
% \end{align}
% is a bijection.
% \end{lemma}
% \begin{proof}
% Let us assume that for some $(i', \tilde{j}', l') \ne (i, \tilde{j}, l)$, $f(i', \tilde{j}', l') = f(i, \tilde{j}, l)$. Then $f(i, \tilde{j}, l) \textnormal{ mod } t = i = f(i', \tilde{j}', l') \textnormal{ mod } t = i'$.
% Similarly, $f(i, \tilde{j}, l) \textnormal{ mod } t(2s-1) = f(i', \tilde{j}', l') \textnormal{ mod } t(2s-1)$ gives $i+t\tilde{j} = i'+t\tilde{j}'$, and thus $j=j'$. Now, because $i=i'$ and $\tilde{j} = \tilde{j}'$ (as we just established), using $f(i, j, l) = f(i', j', l')$, it follows that $l = l'$.
% \end{proof}

\begin{lemma} \label{lem:bijection}
The following function
\begin{align}
\nonumber f: &\{0,\cdots, t-1\} \times \{0, \cdots, 2s-2\} \times \{0,\cdots, t-1\}  \to \{0,\cdots, t^2(2s-1)-1\} \nonumber\\
&(\alpha, \beta, \gamma) \mapsto \alpha + t\beta + t(2s-1)\gamma
\end{align}
is a bijection.
\end{lemma}
\begin{proof}
Let us assume, for the sake of contradiction, that for some $(\alpha', \beta', \gamma') \ne (\alpha, \beta, \gamma)$, $f(\alpha', \beta', \gamma') = f(\alpha, \beta, \gamma)$. Then $(f(\alpha, \beta, \gamma) \textnormal{ mod } t ) = \alpha = ( f(\alpha', \beta', \gamma') \textnormal{ mod } t ) = \alpha'$ and hence $\alpha = \alpha'$.
Similarly, $( f(\alpha, \beta, \gamma) \textnormal{ mod } t(2s-1) ) = ( f(\alpha', \beta', \gamma') \textnormal{ mod } t(2s-1)) $ gives $\alpha+t\beta = \alpha'+t\beta'$, and thus $\beta=\beta'$ (because $\alpha = \alpha'$). Now, because $\alpha=\alpha'$ and $\beta = \beta'$, as we just established,  $f(\alpha, \beta, \gamma) = f(\alpha', \beta', \gamma')$ from our assumption, it follows that $\gamma = \gamma'$. This contradicts our assumption that $(\alpha, \beta, \gamma) \ne (\alpha', \beta', \gamma')$.
\end{proof}

\begin{proof}[Proof of Theorem \ref{thm:tdoff}]
The product of $p_\mathbf{A}(x, y)$ and $p_\mathbf{B}(y,z)$ can be written as follows:
\begin{align}
p_{\mathbf{C}}(x,y,z) &= p_\mathbf{A}(x,y) p_\mathbf{B}(y,z)  \nonumber \\
& = \left(  \sum_{i=0}^{t-1}\sum_{j=0}^{s-1} \mathbf{A}_{i,j} x^{i} y^{j} \right)\left(\sum_{k=0}^{s-1}\sum_{l=0}^{t-1} \mathbf{B}_{k,l} y^{s-1-k} z^{l}\right) \nonumber \\
&= \sum_{i,j,k,l} \mathbf{A}_{i,j} \mathbf{B}_{k,l} x^{i} y^{s-1+j-k} z^{l}.
\end{align}

Note that the coefficient of $x^{i-1} y^{s-1} z^{l-1}$ in $p_{\mathbf{C}}(x,y,z)$  is equal to $\mathbf{C}_{i,l} = \sum_{k=0}^{s-1} \mathbf{A}_{i,k} \mathbf{B}_{k,l}$.
By our choice of $y = x^{t}$ and $z = x^{t(2s-1)}$ we can further simplify $p_{\mathbf{C}}(x, x^t, x^{t(2s-1)})$:
\begin{equation}
p_{\mathbf{C}}(x, y, z) = p_{\mathbf{C}}(x) = \sum_{i,j,k,l} \mathbf{A}_{i,j} \mathbf{B}_{k,l} x^{i + t(s-1+j-k) + t(2s-1)l}. \label{eq:C_poly}
\end{equation}

The maximum degree of this polynomial is when $i = t-1, j-k = s-1$ and $l = t-1$, which is $(t-1)+(2s-2)t+t(2s-1)(t-1) = t^2(2s-1)-1$.
Furthermore, if we let $\alpha = i, \beta = s-1+j-k, \gamma = l$, the function $f(\alpha, \beta, \gamma)$ in Lemma~\ref{lem:bijection} is the degree of $x$ in (\ref{eq:C_poly}). This implies that for different pairs of $(i, j-k, l)$, we get different powers of $x$. When $j-k = 0$, we obtain $(\sum_{k=0}^{s-1} \mathbf{A}_{i,k} \mathbf{B}_{k,l}) x^{i+t(s-1)+t(2s-1)l} = \mathbf{C}_{i,l} x^{i+t(s-1)+t(2s-1)l}$ which is the desired product we want to recover.

This implies that if we have $t^2(2s-1)$ successful worker nodes, we can compute all the coefficients in (\ref{eq:C_poly}) by polynomial interpolation. Hence, we can recover all $\mathbf{C}_{i,l}$'s, \textit{i.e.}, the coefficients of $x^{i+t(s-1)+t(2s-1)l}$, for $i, l = 0,\cdots, t-1$.
\end{proof}

\begin{remark}
\label{rem:recovery_threshold}
PolyDot codes essentially introduce a general framework which transforms the matrix-matrix multiplication problem into a polynomial interpolation problem with three variables $x,y,z$.
For the PolyDot codes proposed in the initial version of this work \cite{allerton17}, we used the substitution $ y=x^t$ and $z=x^{t (2s-1)} $ to convert the polynomial in three variables to a polynomial in a single variable, and obtained the recovery threshold of $t^2(2s-1) $. However, in subsequent works \cite{DNNPaperISIT, entangledpolycodes} following the initial version of this work\cite{allerton17}, by using a different substitution, $x=y^t, z=y^{st}$, the recovery threshold has been improved to $st^2+s-1$ which is an improvement within a factor of $2$. A comparison between different coding strategies for matrix multiplication framework is summarized in Table~\ref{tb:polydot_frame1}. In the table, we clarify how different strategies can be obtained from our general PolyDot framework by different substitutions into $x,y,z$. Polynomial codes \cite{polynomialcodes} and MatDot codes \cite{allerton17} are special cases with $s=1,t=m$ and $s=m,t=1$ respectively.
\begin{table}[h!]
\caption{Comparison of different strategies for multiplying two matrices using different substitutions in the general PolyDot setup.}
\centering
\begin{tabu} to \textwidth { | X[c] | X[c] | X[c] | X[c] | X[c] | }
\hline
{} & Polynomial Code\cite{polynomialcodes}
($st=m$, $s=1$, $t=m$)& MatDot Code \cite{allerton17} ($st=m$, $s=m$, $t=1$) & PolyDot Code \cite{allerton17} ($st=m$) & Generalized PolyDot Code \cite{DNNPaperISIT}, Entangled Poly \cite{entangledpolycodes} ($st=m$) \\ \hline
Substitution &
$x = x, y = x^m, z = x^m$ &
$x = x, y = x, z = x^{2m-1}$ &
$x = x, y = x^t, z = x^{t(2s-1)}$ &
$x = y^t, y= y, z = y^{st}$ \\ \hline
Recovery Threshold &
$m^2 - 1$ &
$2m-1$ &
$t^2(2s-1)$ &
$t^2 s -s -1$ \\ \hline
\end{tabu}
 \label{tb:polydot_frame1}
\end{table}
\end{remark}

\subsection{Complexity analyses of PolyDot codes}\label{sec:complexitytrdoff}
\textbf{Encoding/decoding complexity}:
Encoding a matrix for one worker requires scaling $m$ matrices with $N^2/m$ elements and adding them up. This requires computational complexity of $\mathcal{O}(m \cdot N^2/m)$. Thus, the overall computational complexity of encoding for $P$ worker nodes is $\mathcal{O}(N^2 P)$.

Decoding requires interpolating a polynomial of degree $t^2(2s-1)-1$ for each element in $\mathbf{C}$.
Using polynomial interpolation algorithms of complexity $\mathcal{O}(k \log^2 k)$ \cite{li2000arithmetic, kung1973fast}  where $k=t^2(2s-1)$, the decoding complexity \textit{per matrix element is} $\mathcal{O}(t^2(2s-1) \log^2t^2(2s-1))$. As $\mathbf{C}$ has $N^2$ elements, the overall decoding complexity is  $\mathcal{O}(N^2 t^2(2s-1) \log^2 t^2(2s-1))$

\textbf{Each worker's computational complexity}: Multiplication of matrices of size $N/t \times N/s$ and $N/s \times N/t$ requires $\mathcal{O}(\frac{N^3}{st^2}) = \mathcal{O}(\frac{N^3}{mt})$ computations.

\textbf{Communication complexity}: Master node communicates $\mathcal{O}(N^2/ts) = \mathcal{O}(N^2/m)$ symbols to each worker, hence total outgoing symbols from the master node will be $\mathcal{O}(P N^2/M)$. For decoding, each node sends $\mathcal{O}(N^2/t^2)$ symbols to the fusion node and recovery threshold is $\mathcal{O}(t^2(2s-1))$. Total number of symbols communicated to the fusion node is $\mathcal{O}((2s-1) N^2)$.

\section{Multiplying more than two matrices}
\label{sec:multiple_matrices}
We now present a coding technique for multiplying $n$ matrices (\emph{$n$-matrix multiplication}), \textit{i.e.}, computing $$\mathbf{C} = \mathbf{D}^{(1)} \mathbf{D}^{(2)} \cdots \mathbf{D}^{(n)}.$$
We state the problem formally first and then explain why this is different from multiplying two matrices.
Then we provide a new code construction called \emph{$n$-matrix code} which applies MatDot codes and Polynomial codes in an alternating fashion. With this construction, we show that we can achieve recovery threshold of $\Theta(m^{\lceil n/2 \rceil})$. Later in the section, we describe a generalized $n$-matrix code which has a flexible structure that can trade off communication costs for recovery thresholds.

\subsection{Problem Statement}\label{subsec:nmatprobst}

Compute  the product $\mathbf{C}=\prod_{i=1}^{n} \mathbf{D^{(i)}}$ of  $N\times N$ square matrices, $\mathbf{D^{(1)}}, \cdots, \mathbf{D^{(n)}}$,  in the computational system specified in Section~\ref{sec:model}.
As we will treat odd and even indices of $\mathbf{D}_i$'s differently, we will denote the odd indices of $\mathbf{D^{(i)}}$'s as $\mathbf{A^{(\lceil i/2 \rceil)}}$ and the even indices of $\mathbf{D^{(i)}}$'s as $\mathbf{B^{(i/2)}}$
for all $i \in \{1, \cdots,n\}$. Using this notation, $\mathbf{C}$ can be written as:
    \begin{align}\label{eq:CCC}
\mathbf{C}=\begin{cases} \prod_{i=1}^{\frac{n}{2}}\mathbf{A^{(i)}} \mathbf{B^{(i)}}& \text{ if $n$ is even}  ,\\\left(\prod_{i=1}^{\lfloor{\frac{n}{2}}\rfloor} \mathbf{ A^{(i)}} \mathbf{ B^{(i)}} \right) \mathbf{A^{(\lceil{\frac{n}{2}}\rceil)}} & \text{ if $n$ is odd}. \end{cases}
\end{align}

Each worker can receive at most  $n N^2/m$ symbols from the master node, where each symbol is an element of $\mathbb{F}$.
% We assume that the worker nodes and the master node are allowed to communicate only in the beginning and at the end of the computation.
Similar to Section \ref{sec:problem}, the computational complexities of the operations at master and fusion nodes, in terms of the matrix parameter $N$,  are  required to be strictly less than the computational complexity at any worker node. The goal is to perform this matrix product utilizing faulty or straggling workers with as low recovery threshold as possible. Again, in the following discussion, we will assume that $|\mathbb{F}| > P$.

\begin{remark}
As $n$-matrix multiplication is a chain of $(n-1)$ matrix-matrix multiplications, one may think that we can apply the coding techniques developed in the previous sections to each pairwise matrix multiplication instead of developing a new coding technique for $n$-matrix multiplication. For example, let us consider computing $\mathbf{C} = \mathbf{A}^{(1)} \mathbf{B}^{(1)} \mathbf{A}^{(2)} $. A master node can first encode $\mathbf{A^{(1)}}$ and $\mathbf{B}^{(1)}$ using MatDot codes and distribute encoded matrices to all the worker nodes and the fusion node can decode $\mathbf{E}  = \mathbf{A}^{(1)} \mathbf{B}^{(1)} $ from the output of successful worker nodes. Then we again encode $\mathbf{E}$ and $\mathbf{A}^{(2)}$ using MatDot code and distribute encoded matrices to the worker nodes. Finally, the fusion node can reconstruct $\mathbf{C}$ by decoding the outputs of successful worker nodes. As you can see from this example,  simply applying MatDot codes on each matrix-matrix multiplication requires two rounds of communication after computing $\mathbf{E}  = \mathbf{A}^{(1)} \mathbf{B}^{(1)}$ and $\mathbf{C} = \mathbf{E} \mathbf{A}^{(2)}$. For $n$-matrix multiplication, it requires $n-1$ rounds of communication. This can be inefficient in the systems when the communication cost increases with number of rounds of communication (e.g., due to large communication set up overheads) .

What we propose in this section is a coded $n$-matrix multiplication strategy which requires only one round of communication. Our main result in Theorem~\ref{thm:nmat} shows that $n$-matrix codes need $\Theta(m^{\lceil n/2 \rceil})$ successful nodes to recover the computation result. On the other hand, successively applying MatDot codes requires $\Theta(m)$ nodes to successfully recover the final result, which is is in scaling sense smaller than $\Theta(m^{\lceil n/2 \rceil})$ for large $n$. This suggests that  $n$-matrix codes avoid intermediate communications at the cost of larger recovery threshold. When communication start-up cost is the main source of delay, one should use $n$-matrix code, and when number of computation nodes is limited, one should sequentially apply coding strategy for two-matrix multiplication such as MatDot or PolyDot codes.
\end{remark}

\subsection{Codes for $n$-matrix multiplication}\label{subsec:nmat}
\begin{theorem}[Recovery threshold for multiple matrix multiplications]\label{thm:nmat}
For the matrix multiplication problem specified in Section \ref{subsec:nmatprobst} computed on the system defined in Definition~\ref{def:compsys}, there exists a code with a recovery threshold of
    \begin{align}
    %\hline\nonumber\\
    {k(n,m)=}\left\{ \begin{array}{ll}
        {2m^{{n}/{2}}-1} & \text{if $n$ is even,} \\
        {(m+1)m^{\lfloor\frac{n}{2}\rfloor}-1} &\text{if $n$ is odd.}
        \end{array}\right.\label{expmult}
        %\\
        % \hline\nonumber
        \end{align}
\end{theorem}
\begin{proof}
See Appendix~\ref{app:proof1}.
\end{proof}

Before we describe the general  code construction that will be used to prove Theorem \ref{thm:nmat}, we first present simple examples for even and odd $n$. The first example shows the example for even $n$.

\begin{example} [Multiplying $4$ matrices  ($n=4, m=2, k=7$)] \label{ex:ABCD}

    Here, we give an example of multiplying $4$ matrices and  show that a recovery threshold of $7$ is achievable. For $i\in\{1,2\}$, matrix $\mathbf{A}^{(i)}$ is split vertically into submatrices $\mathbf{A}^{(i)}_0, \mathbf{A}^{(i)}_1$ each of dimension $N \times \frac{N}{2}$ as follows:  $\mathbf{A}^{(i)} =\left[\mathbf{A}^{(i)}_0 \hspace{2pt} \mathbf{A}^{(i)}_1 \right]$,  while, for $i \in\{1,2\}$, matrix $\mathbf{B}^{(i)}$ is split horizontally into submatrices $\mathbf{B}^{(i)}_0, \mathbf{B}^{(i)}_1$ each of dimension $\frac{N}{2} \times N$ as follows:
    \begin{align}\label{eq:exnmat}
\mathbf{B}^{(i)} &= \left[\begin{array}{c} \mathbf{B}^{(i)}_0\\\mathbf{B}^{(i)}_1\end{array}\right].
\end{align}

    Notice that the product $\mathbf{C} = \prod_{i=1}^2 \mathbf{A}^{(i)} \mathbf{B}^{(i)}$ can now be written as
    \begin{align}\label{eq:nmatprod}
        \prod_{i=1}^2 \mathbf{A}^{(i)} \mathbf{B}^{(i)} &=
    \left(\mathbf{A}^{(1)} \mathbf{B}^{(1)} \right)    \left( \mathbf{A}^{(2)} \mathbf{B}^{(2)} \right) = \left( \mathbf{A}^{(1)}_0 \mathbf{B}^{(1)}_0 + \mathbf{A}^{(1)}_1 \mathbf{B}^{(1)}_1 \right)\left( \mathbf{A}^{(2)}_0 \mathbf{B}^{(2)}_0 + \mathbf{A}^{(2)}_1 \mathbf{B}^{(2)}_1\right).
    \end{align}

    Now, we define the encoding polynomials $p_{\mathbf{A}_i}(x), p_{\mathbf{B}_i}(x)$, $i \in \{1,2\}$ as follows:
    \begin{align}\label{eq:nmatenc1}
        p_{\mathbf{A}^{(1)}}(x)&= \mathbf{A}^{(1)}_0 + \mathbf{A}^{(1)}_1 x,\notag\\
        p_{\mathbf{B}^{(1)}}(x)&= \mathbf{B}^{(1)}_0 x + \mathbf{B}^{(1)}_1, \notag\\
        p_{\mathbf{A}^{(2)}}(x)&= \mathbf{A}^{(2)}_0 + \mathbf{A}^{(2)}_1 x, \notag\\
        p_{\mathbf{B}^{(2)}}(x)&= \mathbf{B}^{(2)}_0 x + \mathbf{B}^{(2)}_1.
\end{align}

From (\ref{eq:nmatenc1}), we have
\begin{align}\label{eq:nmatenc}
        p_{\mathbf{A}^{(1)}}(x)& p_{\mathbf{B}^{(1)}}(x) = \mathbf{A}^{(1)}_0 \mathbf{B}^{(1)}_1 + (\mathbf{A}^{(1)}_0 \mathbf{B}^{(1)}_0 + \mathbf{A}^{(1)}_1 \mathbf{B}^{(1)}_1)x+ \mathbf{A}^{(1)}_1 \mathbf{B}^{(1)}_0 x^2, \notag\\
        p_{\mathbf{A}^{(2)}}(x)& p_{\mathbf{B}^{(2)}}(x) = \mathbf{A}^{(2)}_0 \mathbf{B}^{(2)}_1 + (\mathbf{A}^{(2)}_0 \mathbf{B}^{(2)}_0 + \mathbf{A}^{(2)}_1 \mathbf{B}^{(2)}_1)x+ \mathbf{A}^{(2)}_1 \mathbf{B}^{(2)}_0 x^2.
\end{align}

From (\ref{eq:nmatprod}) along with (\ref{eq:nmatenc}), we can observe the following:
\begin{enumerate}[label=(\roman*)]
    \item the coefficient of $x$ in $p_{\mathbf{A}^{(1)}}(x) p_{\mathbf{B}^{(1)}}(x)$ is  $ \mathbf{A}^{(1)}_0 \mathbf{B}^{(1)}_0 + \mathbf{A}^{(1)}_1 \mathbf{B}^{(1)}_1 =\mathbf{A}^{(1)} \mathbf{B}^{(1)}$,
    \item the coefficient of $x^2$ in $p_{\mathbf{A}^{(2)}}(x^2) p_{\mathbf{B}^{(2)}}(x^2)$ is the product $\mathbf{A}^{(2)}_0 \mathbf{B}^{(2)}_0 + \mathbf{A}^{(2)}_1 \mathbf{B}^{(2)}_1 =\mathbf{A}^{(2)} \mathbf{B}^{(2)}$, and
\item the coefficient of $x^3$ in $p_{\mathbf{A}^{(1)}}(x) p_{\mathbf{B}^{(1)}}(x) p_{\mathbf{A}^{(2)}}(x^2) p_{\mathbf{B}^{(2)}}(x^2) $ is the product $\prod_{i=1}^2 \mathbf{A^{(i)}} \mathbf{B^{(i)}}$ (our desired output).
\end{enumerate}

Let $x_1, \cdots, x_{P}$ be distinct elements of $\mathbb{F}$, the master node sends $p_{\mathbf{A^{(i)}}}(x^i_{r})$ and $p_{\mathbf{B^{(i)}}}(x^i_{r})$, for all $i\in\{1,2\}$, to the $r$-th worker node, $r \in \{1, \cdots, P\}$, and the $r$-th worker node performs the multiplication $\prod_{i=1}^2 p_{\mathbf{A^{(i)}}}(x^i_{r}) p_{\mathbf{B^{(i)}}}(x^i_{r})$ and sends the output to the fusion node.

Let worker nodes $1, \cdots, 7$ be the first $7$ worker nodes to send their computation outputs to the fusion node, then the fusion node receives the matrices
$\prod_{i=1}^2 p_{\mathbf{A^{(i)}}}(x^i_{r})p_{\mathbf{B^{(i)}}}(x^i_{r})$ for all  $r \in \{1, \cdots, 7\}$. Since these $7$ matrices can be seen as $7$ evaluations of the matrix polynomial $\prod_{i=1}^2 p_{\mathbf{A^{(i)}}}(x^i) p_{\mathbf{B^{(i)}}} (x^i)$ of degree $6$ at $7$ distinct evaluation points $x_1, \cdots, x_{7}$, the coefficients of the matrix polynomial $\prod_{i=1}^2 p_{\mathbf{A^{(i)}}} (x^i)p_{\mathbf{B^{(i)}}} (x^i)$ can be obtained using polynomial interpolation. This includes the coefficient of $x^3$, \textit{i.e.}, $\prod_{i=1}^2 \mathbf{A^{(i)}} \mathbf{B^{(i)}}$. \hfill $\blacksquare$
\end{example}

Now we show an example for odd $n$.

\begin{example} [Multiplying $3$ matrices ($n =3, m=2, k=5$)] \label{ex:ABC}

Here, we give an example of multiplying $3$ matrices and show that a recovery threshold of $5$ is achievable. In this example, we have three input matrices $\mathbf{A}^{(1)}$, $\mathbf{B}^{(1)}$, and $\mathbf{A}^{(2)}$, each of dimension $N \times N$ and need to compute the product $\mathbf{A}^{(1)} \mathbf{B}^{(1)} \mathbf{A}^{(2)}$. First, the three input matrices are split in the same way as in Example~\ref{ex:ABCD}. The product $\mathbf{A}^{(1)} \mathbf{B}^{(1)} \mathbf{A}^{(2)}$ can now be written as
    \begin{equation}\label{eq:nmatprod2}
\mathbf{C} = \mathbf{A}^{(1)} \mathbf{B}^{(1)} \mathbf{A}^{(2)} = \left[\mathbf{A}^{(1)} \mathbf{B}^{(1)} \mathbf{A}^{(2)}_0  \hspace{3mm} \mathbf{A}^{(1)} \mathbf{B}^{(1)} \mathbf{A}^{(2)}_1 \right],
    \end{equation}
where $\mathbf{A}^{(1)} \mathbf{B}^{(1)}= \mathbf{A}^{(1)}_{0}\mathbf{B}^{(1)}_{0}+ \mathbf{A}^{(1)}_{1} \mathbf{B}^{(1)}_{1}$.

    Now, we define the encoding polynomials $p_{\mathbf{A}_1}(x), p_{\mathbf{B}_1}(x), p_{\mathbf{A}_2}(x)$ as follows:
    \begin{align}
        p_{\mathbf{A}^{(1)}}(x)&= \mathbf{A}^{(1)}_0 + \mathbf{A}^{(1)}_1 x,\notag\\
        p_{\mathbf{B}^{(1)}}(x)&= \mathbf{B}^{(1)}_0 x + \mathbf{B}^{(1)}_1, \notag\\
        p_{\mathbf{A}^{(2)}}(x)&= \mathbf{A}^{(2)}_0 + \mathbf{A}^{(2)}_1 x. \label{eq:nmatenc2}
\end{align}
From (\ref{eq:nmatenc2}), we have

\begin{align} \label{eq:nmatprod3}
        p_{\mathbf{A^{(1)}}}(x) p_{\mathbf{B^{(1)}}}(x) p_{\mathbf{A^{(2)}}}(x^2) &=  \mathbf{A}^{(1)}_0 \mathbf{B}^{(1)}_1 \mathbf{A}^{(2)}_0 +(\mathbf{A}^{(1)}_0 \mathbf{B}^{(1)}_0 + \mathbf{A}^{(1)}_1 \mathbf{B}^{(1)}_1) \mathbf{A}^{(2)}_0 x +(\mathbf{A}^{(1)}_1 \mathbf{B}^{(1)}_0 \mathbf{A}^{(2)}_0 +
    \mathbf{A}^{(1)}_0 \mathbf{B}^{(1)}_1 \mathbf{A}^{(2)}_1)x^2 \notag\\
        &+(\mathbf{A}^{(1)}_0 \mathbf{B}^{(1)}_0 + \mathbf{A}^{(1)}_1 \mathbf{B}^{(1)}_1) \mathbf{A}^{(2)}_1 x^3+ \mathbf{A}^{(1)}_1 \mathbf{B}^{(1)}_0 \mathbf{A}^{(2)}_1 x^4.
\end{align}

From (\ref{eq:nmatprod3}),  we can observe the following:
\begin{enumerate}[label=(\roman*)]
    \item the coefficient of $x$ in $p_{\mathbf{A}^{(1)}}(x) p_{\mathbf{B}^{(1)}}(x) p_{\mathbf{A}^{(2)}}(x^2)$ is the product $\mathbf{A}^{(1)} \mathbf{B}^{(1)} \mathbf{A}^{(2)}_0$, and
\item the coefficient of $x^3$ in $p_{\mathbf{A}^{(1)}}(x) p_{\mathbf{B}^{(1)}}(x) p_{\mathbf{A}^{(2)}}(x^2)$ is the product $\mathbf{A}^{(1)} \mathbf{B}^{(1)} \mathbf{A}^{(2)}_1$.
\end{enumerate}

From (\ref{eq:nmatprod2}), these two coefficients suffice to recover $\mathbf{C}$. Let $x_1, \cdots, x_{P}$ be distinct elements of $\mathbb{F}$, the master node sends $p_{\mathbf{A}^{(i)}}(x^i_{r})$, for all $i\in\{1,2\}$, and $p_{\mathbf{B}_1}(x_{r})$  to the $r$-th worker node, $r \in \{1, \cdots, P\}$, where the $r$-th worker node performs the multiplication $ p_{\mathbf{A}^{(1)}}(x_{r})p_{\mathbf{B}^{(1)}}(x_{r})p_{\mathbf{A}^{(2)}}(x^2_{r})$ and sends the output to the fusion node.

Let worker nodes $1, \cdots, 5$ be the first $5$ worker nodes to send their computation outputs to the fusion node, then the fusion node receives the matrices
$p_{\mathbf{A}^{(1)}}(x_{r})p_{\mathbf{B}^{(1)}}(x_{r})p_{\mathbf{A}^{(2)}}(x^2_{r})$ for all  $r \in \{1, \cdots, 5\}$. Since these $5$ matrices can be seen as $5$ evaluations of the  polynomial $p_{\mathbf{A}^{(1)}}(x)p_{\mathbf{B}^{(1)}}(x)p_{\mathbf{A}^{(2)}}(x^2)$ of degree $4$ at five distinct evaluation points $x_1, \cdots, x_{5}$, the coefficients of the matrix polynomial $p_{\mathbf{A}^{(1)}}(x_{r})p_{\mathbf{B}^{(1)}}(x_{r})p_{\mathbf{A}^{(2)}}(x^2_{r})$ can be obtained using polynomial interpolation. This includes the coefficients of $x$ and $x^3$, \textit{i.e.}, $\mathbf{A}^{(1)} \mathbf{B}^{(1)} \mathbf{A}^{(2)}_0$ and $\mathbf{A}^{(1)} \mathbf{B}^{(1)} \mathbf{A}^{(2)}_1$. \hfill $\blacksquare$
\end{example}

In the following, we present a code construction for $n$-matrix multiplication for general $n$ and $m$.

\begin{construction}\label{con:nmatcod}[An $n$-matrix multiplication code]
\text{}

\textbf{Splitting of input matrices}:
for every $i \in \{1, \cdots, \lceil\frac{n}{2}\rceil\}$ and $j \in \{1, \cdots, \lfloor\frac{n}{2}\rfloor\}$, $\mathbf{A}_i$ and $\mathbf{B}_j$ are split as follows
\begin{equation}
\mathbf{A}^{(i)} = \left[\mathbf{A}^{(i)}_{1} \ \mathbf{A}^{(i)}_{2} \ \ldots \ \mathbf{A}^{(i)}_{m} \right],\;\;\; \mathbf{B}^{(j)} =\left[\begin{array}c \mathbf{B}^{(j)}_{1}\\\mathbf{B}^{(j)}_{2}\\\vdots \\\mathbf{B}^{(j)}_{m}\end{array}\right],
\end{equation}
where, for $k \in \{1,\ldots,m\}$,  $\mathbf{A}^{(i)}_{k},\mathbf{B}^{(j)}_{k}$ are $N \times N/m$ and $N/m \times N$ dimensional matrices, respectively.

\textbf{Master node (encoding)}: Let $x_1, x_2, \ldots, x_{P-1}$ be arbitrary distinct elements of $\mathbb{F}$.  For $i \in \{1, \cdots, \lceil\frac{n}{2}\rceil\}$, define $p_{\mathbf{A}^{(i)}}(x)=\sum_{j=1}^{m}\mathbf{A}^{(i)}_{j}x^{j-1}$, and, for $i \in \{1, \cdots, \lfloor\frac{n}{2}\rfloor\}$, define $p_{\mathbf{B}^{(i)}}(x)= \sum_{j=1}^{m}\mathbf{B}^{(i)}_{j}x^{m-j}.$
%Let $\mathbf{C_{i}}=\mathbf{A_{i}B_{i}}$ and ${p_{\mathbf{C_{i}}}(x)=p_{\mathbf{A_{i}}}(x)p_{\mathbf{B_{i}}}(x)}.$
%The MatDot construction of Section \ref{sec:main} implies that $\mathbf{C_{i}}$ is a co-efficient of $p_{\mathbf{C_i}}(x)$ when $\mathbf{B}_{i} \neq \mathbf{I}.$
For $r\in\{1,2,\dots,P\}$, the master node sends to the $r$-th worker the evaluations, $p_{\mathbf{A}^{(i)}}(x_{r}^{m^{i-1}})$ and $p_{\mathbf{B}^{(j)}}(x_{r}^{m^{j-1}})$, for all   $i\in \{1, \cdots, \lceil \frac{n} {2}\rceil\}$ and $j \in \{1, \cdots, \lfloor \frac{n} {2}\rfloor\}$.

\textbf{Worker nodes}: For $i \in \{1, \cdots, \lceil\frac{n}{2}\rceil\}$,  define
\begin{align}\label{eq:pc}
p_{\mathbf{C}^{(i)}}(x)=\begin{cases}p_{\mathbf{A}^{(i)}}(x)p_{\mathbf{B}^{(i)}}(x)& \text{ if } i\in \{1, \cdots, \lfloor\frac{n}{2}\rfloor\},\\p_{\mathbf{A}^{(i)}}(x)& \text{ if $n$ is odd and } i=\lceil\frac{n}{2}\rceil . \end{cases}
\end{align}
For $r \in \{1,2,\ldots, P\}$, the $r$-th worker node computes the matrix product $\Pi_{i=1}^{\lceil\frac{n}{2}\rceil}p_{\mathbf{C}^{(i)}}(x_{r}^{m^{i-1}})$ and sends it to the fusion node on successful completion.

\textbf{Fusion node (decoding)}:
If $n$ is even, the fusion node  uses outputs of any $2m^{\frac{n}{2}}-1$ successful workers to  compute the coefficient of  $x^{m^{ {n}/{2}}-1}$ in  the matrix polynomial $\Pi_{i=1}^{\frac{n}{2}}p_{\mathbf{C}^{(i)}}(x^{m^{i-1}})$, and if $n$ is odd, the fusion node  uses outputs of any $m^{\lfloor \frac{n}{2} \rfloor}(m+1)-1$ successful workers to  compute the coefficients  of $x^{j m^{\lfloor \frac{n}{2}\rfloor}-1}$, for all $j \in \{1,\cdots,m\}$,  in the matrix polynomial $\Pi_{i=1}^{\lceil\frac{n}{2}\rceil} p_{\mathbf{C}^{(i)}}(x^{m^{i-1}})$  (the feasibility of this step will be shown later in the proof of Theorem \ref{thm:nmat}).

If the number of successful workers is smaller than $2m^{\frac{n}{2}}-1$ for even $n$ or smaller than  $m^{\lfloor \frac{n}{2} \rfloor}(m+1)-1$ for odd $n$, the fusion node declares a failure.
\end{construction}
\begin{remark}\label{rmk:nmat}
The coefficient of $x^{m^i-m^{i-1}}$ in $p_{\mathbf{C}_i}(x^{m^{i-1}})$, for any $i \in \{1, \cdots, \lfloor \frac{n}{2} \rfloor\}$, is $\sum_{j=1}^{m} \mathbf{A}^{(i)}_{j}\mathbf{B}^{(i)}_{j}=\mathbf{A}^{(i)} \mathbf{B}^{(i)}$.
\end{remark}

\subsection{Complexity analyses of Construction \ref{con:nmatcod}}\label{sec:nmatcomplexity}
\textbf{Encoding/decoding complexity}:  Decoding requires interpolating a $2m^{n/2}-2$ degree polynomial if $n$ is even or a $m^{\lfloor\frac{n}{2}\rfloor}(m+1)-2$ degree polynomial if $n$ is odd for each element in the matrix. Using polynomial interpolation algorithms of complexity $\mathcal{O}(k \log^2 k)$~\cite{kung1973fast},  where $k=k(n,m)$, defined in (\ref{expmult}), complexity per matrix element is $\mathcal{O}(m^{\lceil\frac{n}{2}\rceil} \log^2 m^{\lceil\frac{n}{2}\rceil})$. Thus, for $N^2$ elements, the decoding complexity is  $\mathcal{O}(N^2 m^{\lceil\frac{n}{2}\rceil} \log^2 m^{\lceil\frac{n}{2}\rceil}).$

Encoding for each worker requires performing $n$ additions, each adding $m$ scaled matrices of size $N^2/m$, for an overall encoding complexity for \textit{each worker} of $\mathcal{O}(mnN^2/m)=\mathcal{O}(nN^2)$. Thus, the overall computational complexity of encoding for $P$ workers is $\mathcal{O}(nN^2P)$.

%\textbf{Each worker's computational cost}: Each worker multiplies $n$ matrices of  dimensions $N\times N/m$ and $N/m\times N$. This requires $n-1$ matrix multiplications, each of them is either a multiplication of two matrices of dimensions $N \times N/m$ and $N/m \times N$ requiring $N^3/m$ operations or a multiplication of two matrices of dimensions $N \times N$ and $N \times N/m$ requiring $N^3/m$ operations as well. Thus, the overall matrix multiplication of $n$ matrices requires $(n-1)N^3/m$ operations, and the computational complexity for each worker is $\mathcal{O}((n-1)N^3/m)$.
%Thus, as long as $P\ll N$ (and hence $m\ll N$), encoding and decoding complexity is much smaller than per-worker complexity.

\textbf{Each worker's computational cost}: Each worker multiplies $n$ matrices of  dimensions $N\times N/m$ and $N/m\times N$. For any worker $r$ with $r \in \{1, \cdots, P\}$, the multiplication can be performed as follows:

\textbf{Case $1$:} $n$ is even

In this case, worker $r$ wishes to compute $p_{\mathbf{A}^{(1)}}(x_{r})p_{\mathbf{B}^{(1)}}(x_{r})p_{\mathbf{A}^{(2)}}(x_{r}^m)p_{\mathbf{B}^{(2)}}(x_{r}^m) \cdots p_{\mathbf{A}^{(n/2)}}(x_{r}^{m^{n/2-1}}) p_{\mathbf{B}^{(n/2)}}(x_{r}^{m^{n/2-1}})$. Worker $r$ does this multiplication in the following order:
\begin{itemize}
    \item[1.] Compute $p_{\mathbf{B}^{(i)}}(x^{m^{i-1}}_{r})p_{\mathbf{A}^{(i+1)}}(x_{r}^{m^i})$ for all $i \in \{1, \cdots, n/2-1\}$ with a total complexity of $\mathcal{O}(nN^3/m^2)$.
    \item[2.] Compute the product of the output matrices of the previous step with a total complexity of $\mathcal{O}(nN^3/m^3)$. Call this product matrix $\mathbf{D}$. Notice that $\mathbf{D}$ has a dimension of $N/m \times N/m$.
    \item[3.] Compute $p_{\mathbf{A}^{(1)}}(x_{r}) \mathbf{D}$ with complexity $\mathcal{O}(N^3/m^2)$. Call this product matrix $\mathbf{E}$. Notice that $\mathbf{E}$ has a dimension of $N \times N/m$.
    \item[4.] Compute $\mathbf{E}\hspace{1mm} p_{\mathbf{B}^{(n/2)}}(x_{r}^{m^{n/2-1}})$ with complexity $\mathcal{O}(N^3/m)$.
\end{itemize}

Hence, the overall computational complexity per worker for even $n$ is $\mathcal{O}(\max(nN^3/m^2, nN^3/m^3, N^3/m^2, N^3/m))=\mathcal{O}(\max(nN^3/m^2, N^3/m))$.

\textbf{Case $2$:} $n$ is odd

In this case, worker $r$ wishes to compute $$p_{\mathbf{A}^{(1)}}(x_{r})p_{\mathbf{B}^{(1)}}(x_{r})p_{\mathbf{A}^{(2)}}(x_{r}^m)p_{\mathbf{B}^{(2)}}(x_{r}^m) \cdots p_{\mathbf{A}^{((n-1)/2)}}(x_{r}^{m^{(n-3)/2}}) p_{\mathbf{B}^{((n-1)/2)}}(x_{r}^{m^{(n-3)/2}})p_{\mathbf{A}^{((n+1)/2)}}(x_{r}^{m^{(n-1)/2}}).$$ Worker $r$ does this multiplication in the following order:
\begin{itemize}
    \item[1.] Compute $p_{\mathbf{B}^{(i)}}(x^{m^{i-1}}_{r})p_{\mathbf{A}^{(i+1)}}(x_{r}^{m^i})$ for all $i \in \{1, \cdots, (n-1)/2\}$ with a total complexity of $\mathcal{O}(nN^3/m^2)$.
    \item[2.] Compute the product of the output matrices of the previous step with a total complexity of $\mathcal{O}(nN^3/m^3)$. Call this product matrix $\mathbf{D}$. Notice that $\mathbf{D}$ has a dimension of $N/m \times N/m$.
    \item[3.] Compute $p_{\mathbf{A}^{(1)}}(x_{r}) \mathbf{D}$ with complexity $\mathcal{O}(N^3/m^2)$.
\end{itemize}

Hence, the overall computational complexity per worker for odd $n$ is $\mathcal{O}(\max(nN^3/m^2, nN^3/m^3, N^3/m^2))=\mathcal{O}(nN^3/m^2)$.

In conclusion, the  computational complexity per worker is $\mathcal{O}(\max(nN^3/m^2, N^3/m))$ if $n$ is even, and $\mathcal{O}(nN^3/m^2)$ if $n$ is odd\footnote{The expressions for even $n$ and odd $n$ are different due to the last step in the even $n$ case where we compute the matrix multiplication of dimension $N\times N/m$ and $N/m \times N$, which has computational complexity of $\mathcal{O}(N^3/m)$}.

%This requires $n-1$ matrix multiplications, each of them is either a multiplication of two matrices of dimensions $N \times N/m$ and $N/m \times N$ requiring $N^3/m$ operations or a multiplication of two matrices of dimensions $N \times N$ and $N \times N/m$ requiring $N^3/m$ operations as well. Thus, the overall matrix multiplication of $n$ matrices requires $(n-1)N^3/m$ operations, and the computational complexity for each worker is $\mathcal{O}((n-1)N^3/m)$.

%\textcolor{red}{Minor point: I think computational complexity here is $\mathcal{O}(nN^3/\mathbf{m^2})$ (Correct me if I'm wrong). \\If we multiply matrices of sizes $N/m \times N$ and $N \times N/m$ first, \textit{i.e.}, multiplying $B_1$ and  $A_2$, this multiplication requires  $\mathcal{O}(N^3/m^2)$ operations. We can do these multiplications for the middle matrices and obtain $B_1 A_2, \cdots, B_{n/2-1}A_{n/2}$. Computational complexity for this is $\mathcal{O}(\lfloor n-1/2 \rfloor N^3 /m^2)$. We then keep multiplying these middle matrices of size $N/m \times N/m$: $\mathcal{O}(N^3/m^3)$ until we get $B_1 A_2 \cdots A_{n/2}$. In the last steps, we multiply $A_1$ and $B_{n/2}$ (assuming that $n$ is even): $A_1 \cdot (B_1 A_2 \cdots A_{n/2}) \cdot B_{n/2}$. This requires one multiplication of complexity $\mathcal{O}(N^3/m^2)$ and one multiplication of complexity $\mathcal{O}(N^3/m)$. Overall, computational complexity is $\mathcal{O}(n/2 N^3 /m^2 + (n/2-2) N^3 /m^2 + N^3/m)$. For $n$ odd, it would be $\mathcal{O}((n+1)/2 N^3/m^2 +(n-3)/2 N^3/m^3)$.}

\textbf{Communication cost}:
The master node communicates total of $\mathcal{O}(nPN^2/m)$ symbols to the worker nodes, and the fusion node receives $\mathcal{O}(m^{\lfloor\frac{n}{2}\rfloor}N^2)$ symbols from the successful worker nodes. %\textcolor{red}{Again minor point: this should be $\mathcal{O}(m^{\lfloor \frac{n}{2} \rfloor}N^2)$ because when $n$ is odd, each worker node sends $N^2/m$ symbols instead of $N^2$ symbols.}

\subsection{Generalized $n$-matrix multiplication}
Here, we give a new code construction which is a generalization of the code construction given in the previous section. The new construction let us split input matrices more flexibly and trade off communication and computation (similar to Section~\ref{sec:polydot} for two matrices). This result is an improvement on \cite{allerton17} and builds on techniques from~\cite{SanghamitraISIT2017,entangledpolycodes}.

\begin{theorem}
[Generalized recovery threshold for multiple matrix multiplications]\label{thm:nmat2}
For the matrix multiplication problem specified in Section \ref{subsec:nmatprobst} computed on the system defined in Definition~\ref{def:compsys}, there exists a code with a recovery threshold of
    \begin{align}
    %\hline\nonumber\\
    {k(n,s,t)=}\left\{ \begin{array}{ll}
        {s^{\frac{n}{2}}t^{\frac{n}{2}+1} + s^{\frac{n}{2}}t^{\frac{n}{2}}-t} & \text{if $n$ is even,} \\
        {s^{\frac{n+1}{2}}t^{\frac{n+1}{2}}+s^{\frac{n-1}{2}} t^{\frac{n+1}{2}}-t} &\text{if $n$ is odd}
        \end{array}\right.\label{expmult2}
        %\\
        % \hline\nonumber
    \end{align}
for any integers $s, t$ that satisfy $m = st$.
\end{theorem}
\begin{proof}
See Appendix~\ref{app:proof2}.
\end{proof}

\begin{remark}
If we substitute $st = m$ in (\ref{expmult}), we get:
\begin{align}
%\hline\nonumber\\
    {k(n,s,t)=}\left\{ \begin{array}{ll}
    {m^{\frac{n}{2}}(t+1)-t} & \text{if $n$ is even,} \\
    {m^{\frac{n-1}{2}}(m+t)-t} &\text{if $n$ is odd}
    \end{array} \right.\label{expmult3}
    %\\
    % \hline\nonumber
\end{align}

By plugging in $s = m, t = 1$, we can see that $k(n,s,t) = 2m^{n/2}-1$ for $n$ even, and $k(n,s,t) = m^{\frac{n+1}{2}} + m^{\frac{n-1}{2}} -1$ for $n$ odd. This matches the recovery threshold given in (\ref{expmult}).
\end{remark}

We now give a construction of generalized $n$-matrix code.

\begin{construction}[Generalized $n$-matrix multiplication code] \label{const:nmatrix2}
\text{}

\textbf{Splitting of input matrices}: We split $A_i$'s and $B_i$'s as follows:
\begin{align}\label{eq:AiBi}
\mathbf{A}^{(i)} = \left[\begin{array}{ccc} \mathbf{A}^{(i)}_{0,0} & \cdots &\mathbf{A}^{(i)}_{0,s-1}\\\vdots & \ddots & \vdots \\\mathbf{A}^{(i)}_{t-1,0} &\cdots& \mathbf{A}^{(i)}_{t-1,s-1}\end{array}\right], \quad \mathbf{B}^{(i)} = \left[\begin{array}{ccc} \mathbf{B}^{(i)}_{0,0} & \cdots &\mathbf{B}^{(i)}_{0,t-1}\\\vdots & \ddots & \vdots \\\mathbf{B}^{(i)}_{s-1,0}&\cdots& \mathbf{B}^{(i)}_{s-1,t-1}\end{array}\right],
\end{align}
where $A^{(i)}_{j,k}$'s have dimension $N/t \times N/s$ and $B^{(i)}_{j,k}$'s have dimension $N/s \times N/t$.

\textbf{Master node (encoding)}: Define the encoding polynomials as
\begin{align*}
p_{\mathbf{A}^{(1)}}(z_1, z_2) &= \sum_{i=0}^{t-1}\sum_{j=0}^{s-1} \mathbf{A}^{(1)}_{i,j} z_1^{i} z_2^{j}, \\
p_{\mathbf{B}^{(1)}}(z_2, z_3) &= \sum_{i=0}^{s-1}\sum_{j=0}^{t-1} \mathbf{B}^{(1)}_{i,j} z_2^{s-1-i} z_3^{j}, \\
&\vdots, \\
p_{\mathbf{B}^{(n/2)}}(z_n, z_{n+1}) &= \sum_{i=0}^{s-1}\sum_{j=0}^{t-1} \mathbf{B}^{(n/2)}_{i,j} z_n^{s-1-i} z_{n+1}^{j}.
\end{align*}
for $n$ even, and
\begin{align*}
p_{\mathbf{A}^{(1)}}(z_1, z_2) &= \sum_{i=0}^{t-1}\sum_{j=0}^{s-1} \mathbf{A}^{(1)}_{i,j} z_1^{i} z_2^{j}, \\
&\vdots, \\
p_{\mathbf{B}^{((n-1)/2)}}(z_{n-1}, z_n) &= \sum_{i=0}^{s-1}\sum_{j=0}^{t-1} \mathbf{B}^{((n-1)/2)}_{i,j} z_{n-1}^{s-1-i} z_n^{j}, \\
p_{\mathbf{A}^{((n+1)/2)}}(z_n, z_{n+1}) &= \sum_{i=0}^{t-1}\sum_{j=0}^{s-1} \mathbf{A}^{((n-1)/2)}_{i,j} z_n^{t-1-i} z_{n+1}^{j},
\end{align*}
for $n$ odd.

The master node sends to the $r$-th worker evaluations of $p_{\mathbf{A}^{(i)}}$'s, and $p_{\mathbf{B}^{(i)}}$'s at $z_1=x_r, z_2= x_r^t, z_3= x_r^{st}, z_4 = x_r^{st^2}, z_5 = z_r^{s^2 t^2}, \cdots, z_{n+1} = x_r^{s^{\lfloor n/2 \rfloor} t^{\lceil n/2 \rceil}}$. $x_r$'s are all distinct for $r \in \{1,2,\ldots, P\}$.

\textbf{Worker nodes}:  For $r \in \{1,2,\ldots, P\}$, the $r$-th worker node computes the matrix product $p_{\mathbf{C}}(x_r)= \prod_i p_{\mathbf{A}^{(i)}}(x_r) p_{\mathbf{B}^{(i)}}(x_r)$ and sends it to the fusion node on successful completion.

\textbf{Fusion node (decoding)}: The fusion node  uses outputs of any $k(n,s,t)$ successful workers (given in (\ref{expmult2})) to  compute the coefficients of $p_{\mathbf{C}}(z)$. If the number of successful workers is smaller than $k(n,s,t)$, the fusion node declares a failure.
\end{construction}

We can further optimize the result in Theorem~\ref{thm:nmat2} by choosing a different substitution as below:
\begin{align}
z_1 = x^{s^{n/2}t^{n/2-1}}, z_2 = x, z_3 = x^{s}, \cdots, z_{n} = x^{s^{n/2-1}t^{n/2-1}}, z_{n+1} = x^{s^{n/2}t^{n/2}} &\textnormal{ for } n \textnormal{ even,} \\
z_1 = x^{s^{(n-1)/2}t^{(n-1)/2}}, z_2 = x, x_3 = x^{s}, \cdots, z_{n} = x^{s^{(n-1)/2}t^{(n-3)/2}}, z_{n+1} = x^{s^{(n-1)/2}t^{(n+1)/2}} &\textnormal{ for } n \textnormal{ odd.}
\end{align}
The main intuition behind this substitution is that for $z_1$ and $z_{n+1}$, their powers grow from $0$ to $t-1$ (or $s-1$), while all the other terms have powers growing from $0$ to $2s-2$ (or $2t-2$). Hence, to minimize the maximum degree of the product polynomial, it is best to assign high powers of $x$ to $z_1$ and $z_{n+1}$. This assignment gives us the improved recovery threshold:
\begin{theorem}
[Improved recovery threshold for multiple matrix multiplications]\label{thm:nmat3}
For the matrix multiplication problem specified in Section \ref{subsec:nmatprobst} computed on the system defined in Definition~\ref{def:compsys}, there exists a code with a recovery threshold of
    \begin{align}
    %\hline\nonumber\\
    {k(n,s,t)=}\left\{ \begin{array}{ll}
        {s^{\frac{n}{2}}t^{\frac{n}{2}+1} + s^{\frac{n}{2}}t^{\frac{n}{2}-1}-1} & \text{if $n$ is even,} \\
        {s^{\frac{n+1}{2}}t^{\frac{n+1}{2}}+s^{\frac{n-1}{2}} t^{\frac{n-1}{2}}-1} &\text{if $n$ is odd}
        \end{array}\right.\label{eq:nmat3}
        %\\
        % \hline\nonumber
    \end{align}
for any integers $s, t$ that satisfy $m = st$.
\end{theorem}

\begin{remark}
    In this paper, we present three different strategies for multiplying $n$ matrices. All of them can be understood in our general PolyDot setup -- they all differ in the substitutions for the variables $x_1, \cdots, x_{n+1}$ to convert the polynomial in $n$ variables into a polynomial in a single variable (See Table~\ref{tb:polydot_frame2}) for the ease of interpolation.

\begin{table}[h!]
\caption{Comparison of different strategies for multiplying $n$ matrices using different substitutions in the general PolyDot setup when $n$ is even. }
\centering
\begin{tabu} to 0.8\textwidth { | X[c] | X[c] | X[c] | X[c] | }
\hline
{} & $n$-matrix code & Generalized $n$-matrix code & Improved $n$-matrix code \\ \hline
Substitution &
$z_1 = z_2 = x, z_3 = z_4 = x^{m}, \cdots, z_{n-1} = x_{n} = x^{m^{n/2-1}}, z_{n+1} =  x^{m^{n/2}}$ &
$z_1 = x, z_2 = x^t, z_3 = x^{st}, \cdots, z_{n+1} = x^{s^{n/2} t^{n/2}}$ &
$z_1 = x^{s^{n/2}t^{n/2-1}}, z_2 = x, x_3 = x^{s}, \cdots, z_{n} = x^{s^{n/2-1}t^{n/2-1}}, z_{n+1} = x^{s^{n/2}t^{n/2}}$ \\ \hline
Recovery Threshold &
$2m^{n/2}-1$ &
$s^{\frac{n}{2}}t^{\frac{n}{2}+1} + s^{\frac{n}{2}}t^{\frac{n}{2}}-t$ &
$s^{\frac{n}{2}}t^{\frac{n}{2}+1} + s^{\frac{n}{2}}t^{\frac{n}{2}-1}-1$ \\ \hline
\end{tabu}
    \label{tb:polydot_frame2}
\end{table}

Also, note that if we consider the substitutions given here for the particular case of two-matrix multiplication, \textit{i.e.}, $n=2$, the recovery thresholds that we obtain are actually better than the recovery thresholds proposed in the initial version of this work in \cite{allerton17}, and match with the recovery threshold for two matrix multiplication in the subsequent works \cite{entangledpolycodes,DNNPaperISIT}.

\end{remark}

\subsection{Complexity analyses of generalized $n$-matrix codes}
\textbf{Encoding/decoding complexity}: Encoding communication cost is $\mathcal{O}(nN^2P)$ as in Section~\ref{sec:nmatcomplexity}. Decoding complexity is $\mathcal{O}(N^2 k(n,s,t) \log^2 k(n,s,t) )$.

\textbf{Communication Complexity}: The master node sends out $\mathcal{O}(nPN^2/m)$ encoded symbols in the beginning.
After the completion of computation, each node has to send $\mathcal{O}(N^2/t^2)$ symbols to the fusion node. Hence, total number of symbols the fusion node receives is $k(n,s,t) \cdot N^2 /t^2$.
Let us first consider the case when $n$ is even. By substituting (\ref{eq:nmat3}), we obtain $k(n,s,t) N^2 /t^2 = \mathcal{O}(m^{n/2}/t)$. This is the same trade-off we observed using PolyDot codes for single matrix-matrix multiplication. For a fixed $m$, recovery threshold $k(n, s,t)$ grows linearly with $t$ while communication cost is inversely related to $t$ (See Fig.\ref{fig:tradeoff}).
When $n$ is odd, we do not see such trade-off.
Recovery threshold is always $m^{(n-1)/2}(m+t)-t = \mathcal{O}(m^{(n+1)/2})$ regardless of the choice of $t$.
Communication cost on the other hand is $k(n,s,t) N^2 /t^2 = \mathcal{O}(m^{(n+1)/2}/t^2 +m^{(n-1)/2}/t)$ which decreases with growing $t$. For instance, if $t=1$, communication cost is $\mathcal{O}(m^{(n+1)/2})$, and when $t = m$, communication cost is $\mathcal{O}(m^{(n-3)/2})$. This suggests that when $n$ is odd, it is always better to choose $t=m$ as $m$ grows to infinity.

\textbf{Each worker's computation cost}: Using the similar technique shown in Section~\ref{sec:nmatcomplexity}, we can show that each worker's computation complexity is at most $\mathcal{O}(\max (nN^3/m^{1.5}, N^3/m))$ for any choice of $s, t$. If we compare the computation complexity for encoding/decoding and the computation complexity at each worker node, we can see that as long as $N > \mathcal{O}(m^{n/2-1.5}\log m)$, encoding/decoding computation overhead is amortized.

\begin{remark}
Our result given here splits $\mathbf{A}^{(i)}$'s in to $s \times t$ grid of blocks and $\mathbf{B}^{(i)}$'s into $t \times s$ grid of blocks. However, it is not necessary that all matrices have to be split in the same fashion. For instance, $\mathbf{A}^{(1)}$ can be divided into $t_1 \times s_1$ grid and $\mathbf{B}^{(1)}$ can be divided into $s_1 \times t_2$ grid, and so on. In this more general setting $\mathbf{A}^{(i)}$'s are split into $t_i \times s_i$ grid and $\mathbf{B}^{(i)}$'s are split into $s_i \times t_{i+1}$ grid. Let us denote $\mathbf{s} = [s_1, \cdots, s_{n/2}], \mathbf{t} = [t_1, \cdots, t_{n/2+1}]$. Then Theorem~\ref{thm:nmat3} can be rewritten as follows.

\begin{equation}
k(n, \mathbf{s}, \mathbf{t}) = \begin{cases} (t_{n/2+1}-1/t_1) \prod_{i=1}^{n/2} s_i  t_i - 1 \textnormal{ if } n \textnormal{ even, } \\
(t_1 s_{(n+1)/2} + 1) \prod_{i=1}^{(n-1)/2} s_i t_i - 1 \textnormal{ if } n \textnormal{ odd. }
\end{cases}
\end{equation}

In this work we assumed that all matrices have size $N \times N$ for simplicity. However, this assumption is not necessary in  the results presented here. When we have matrices with different dimensions to multiply, splitting each matrix in a different way would be more beneficial. For example, when we multiply matrices $\mathbf{A}, \mathbf{B}$ with dimensions $N \times N$ and $N \times 2$, we can divide $\mathbf{A}$ into $t \times s$ grid and divide $\mathbf{B}$ into $s \times 1$ grid.
\end{remark}

\section{Discussion and Future Work}
\label{sec:discussion}
We provide a novel MatDot code construction for coded matrix multiplication with a recovery threshold of $2m-1$. We also present the systematic MatDot construction achieving the same recovery threshold of $2m-1$. A recent converse of Yu et al.~\cite{entangledpolycodes} shows that the recovery thresholds of MatDot codes are in fact optimal. This paper also provides full proofs of results that appeared in~\cite{allerton17}, including PolyDot constructions which allow a trade-off between communication cost and recovery threshold. Finally, we provide code constructions for multiplying more than two matrices.

We conclude with a discussion that uses an important open problem, namely coded tensor products, to demonstrate how focusing exclusively on recovery thresholds, and ignoring encoding/decoding costs in coded computing problems, can provide impractical solutions.

\subsection{When is coded computing useful? An example of coded tensor products}

 Consider the problem of computing the tensor product of two $N\times N$ square matrices $\mathbf{A}$ and $\mathbf{B}$, \textit{i.e.}, $\mathbf{A}\otimes \mathbf{B}$, using $P$ workers in the system defined in Section~\ref{sec:sysmod}. As usual, our goal is to implement this in a parallelized fashion with a low recovery threshold. For this problem, we show (below) that an application of Polynomial codes \cite{polynomialcodes} yields a recovery threshold of $m^{2}$. However, we also show that this makes the decoding complexity at the fusion node comparable to (or sometimes even larger than) the overall per-worker computational complexity. This can be undesirable when coded computing is performed to address straggling because now the fusion node itself becomes the bottleneck. This leads to two interesting questions for future work:
\begin{itemize}
\item Is there an application where the high decoding cost at the fusion node can be justified?
\item Are there alternative techniques of coding tensor products with reduced decoding overhead?
\end{itemize}

%We also include our proposed solution here for completeness.

To be concrete, the Polynomial coded tensor-product strategy is as follows. We split the two matrices $\mathbf{A}$ and $\mathbf{B}$ as follows:
$$\mathbf{A}=\begin{bmatrix}\mathbf{A}_0 & \mathbf{A}_1& \dots& \mathbf{A}_{m-1}\end{bmatrix},\qquad \mathbf{B}=\begin{bmatrix}\mathbf{B}_0 & \mathbf{B}_1 &\dots &\mathbf{B}_{m-1}\end{bmatrix}.$$
Note that
\begin{align*}
\mathbf{A} \otimes \mathbf{B} &= [\mathbf{A}_0 \otimes \mathbf{B} ~~ \ldots  ~~ \mathbf{A}_{m-1} \otimes \mathbf{B}] \\
&= [\mathbf{A}_0 \otimes \mathbf{B}_0\;\; \mathbf{A}_0 \otimes \mathbf{B}_1 \;  \cdots \mathbf{A}_{0} \otimes \mathbf{B}_{m-1} ~~ \cdots ~~ \mathbf{A}_{m-1} \otimes \mathbf{B}_{0}  \; \cdots \;   \mathbf{A}_{m-1} \otimes \mathbf{B}_{m-1}  ]
\end{align*}
and thus it suffices to compute all terms of the form $\mathbf{A}_i \otimes \mathbf{B}_j$ for $i,j = 0, \cdots, m-1$ to obtain $\mathbf{A} \otimes \mathbf{B}$.

Let us define $p_\mathbf{A}(x) = \sum_{i=0}^{m-1} \mathbf{A}_i x^{i}$ and $p_\mathbf{B}(x) = \sum_{j=0}^{m-1} \mathbf{B}_j x^{mj}$ respectively.
Let us also choose distinct scalars $x_1, x_2, \ldots, x_{P}$ for each worker. Each worker receives the evaluation of $p_\mathbf{A}(x)$ and $p_\mathbf{B}(x)$ at distinct scalar values, \textit{i.e.}, at $x=x_1,x_2,\ldots,x_P$ respectively. The worker then computes the tensor product $p_\mathbf{A}(x) \otimes p_\mathbf{B}(x)$ that we will denote as $p_{\mathbf{A} \otimes \mathbf{B}}(x)$ at a distinct scalar value of $x$. Thus,  worker $r$ computes $p_\mathbf{A}(x_r) \otimes p_\mathbf{B}(x_r)$ for $r=1,2,\ldots,P$.

Observe that $p_{\mathbf{A} \otimes \mathbf{B}}(x)$ is a polynomial of degree $m^2-1$.
$$p_{\mathbf{A} \otimes \mathbf{B}}(x)= p_\mathbf{A}(x) \otimes p_\mathbf{B}(x) =  \left( \sum_{i=0}^{m-1} \mathbf{A}_i x^{i} \right) \otimes \left( \sum_{j=0}^{m-1} \mathbf{B}_j x^{mj} \right) = \sum_{i=0}^{m-1}\sum_{j=0}^{m-1} (\mathbf{A}_i \otimes \mathbf{B}_j ) x^{i+mj} .$$
The coefficient of $x^{i+mj}$ in $p_{\mathbf{A} \otimes \mathbf{B}}(x)$ is in fact $\mathbf{A}_{i} \otimes \mathbf{B}_{j},$ for {$0 \leq i,j \leq m-1$}. Thus, if the fusion node is able to interpolate all the coefficients of the polynomial $p_{\mathbf{A} \otimes \mathbf{B}}(x)$, it can successfully recover all the matrices in the set $\{\mathbf{A}_{i}\otimes \mathbf{B}_{j}, i,j=0,1,\ldots, m-1\},$ and therefore $\mathbf{A} \otimes \mathbf{B}.$ Because the polynomial is of degree $m^2-1$, the fusion node needs $m^2$ evaluations of the polynomial at distinct values. Worker node $r$ produces an evaluation of the polynomial $p_{\mathbf{A} \otimes \mathbf{B}}(x)$ at $x=x_r$. The fusion node is thus required to wait for any $m^{2}$ successful worker nodes, and then it can interpolate $p_{\mathbf{A}\otimes \mathbf{B}}(x)$.

%While the above approach obtains a recovery threshold of $m^{2}$, the workers do not necessarily carry a greater burden of computation as compared with the fusion node as is desirable. Concretely,

The computational complexity of the fusion node is $\Theta\left(\frac{N^4}{m^2}m^2 \log^2 (m^2) \right) = \Theta\left(N^4 \log^2 (m) \right),$ which is of the same order as the complexity of the entire tensor product, \textit{i.e.}, $\Theta(N^4)$, and thus, in scaling sense, same or higher than per-processor complexity. Ideally, we would like the computational complexity at the fusion node to be negligible in a scaling sense as compared to the computational complexity at each worker node. %Finding an application where high decoding cost can be justified in the aforementioned coded tensor product algorithm is an interesting open question.

%\subsection{Other Applications of coded computing}

%Coded computing has been applied to several linear operations such as, matrix-vector products, matrix-matrix products, convolutions, Fourier transforms, Deep Neural Networks, and so on. It will be an interesting future work to understand if there are other applications to which these ideas extend.

More generally, let $g_e(P)$ and $g_d(P)$ denote upper bounds (may be loose but are functions of $P$ alone) on the respective encoding and decoding complexities of the code that we are choosing when the length of the codewords are $P$. To determine whether coded computing is a viable option for a computation which can be coded, we can conceptually classify computations that can be coded into different categories as follows:
\begin{itemize}
\item[(a)] Size of the inputs and outputs are negligible in scaling sense compared to the overall computational complexity: some examples are convolutions (Input Size, Output Size $=\Theta(N)$, Computational Complexity $=\Theta(N\log_2{N})$) and matrix-matrix products (Input Size, Output Size $=\Theta(N^2)$, Computational Complexity$=\Theta(N^3)$). Note that, existing strategies (including the ones in this paper) encode groups of symbols of the size of only a fraction of the input, e.g., row-blocks or column-blocks of input matrices $\mathbf{A}$ and $ \mathbf{B}$. For computations of this category, the encoding complexity is \textbf{at most} $\mathcal{O}(\text{Input Size}\times  g_e(P))$ which is the complexity of encoding groups of symbols of the size of the whole input. The communication from the master is also at most $\mathcal{O}(\text{Input Size}\times P)$ which is the cost of sending the whole input to all the $P$ workers. Thus, the encoding complexity and initial communication complexity can be made negligible as compared to the per-worker computational complexity if the number of workers, \textit{i.e.}, $P$ satisfies the following condition (may be loose): $$ \text{Input Size} \times(g_e(P)+P) = o \left( \frac{\text{Computational Complexity}}{ P } \right).$$
Similarly, the decoding complexity is at most $\mathcal{O}(\text{Output Size}\times  g_d(P))$ which is the cost of decoding once for every symbol of the output and the communication to the fusion node is at most $\mathcal{O}(\text{Output Size} \times P)$. Thus, the decoding complexity can also be made negligible if $P$ satisfies: $$ \text{Output Size} \times (g_d(P)+P) = o \left( \frac{\text{Computational Complexity}}{ P } \right). $$
As an example, for coding matrix-matrix products, we require $P g_e(P), \ Pg_d(P)=o(N)$ for encoding and decoding to be done online with negligible overhead as compared to the per-node computational complexity. For convolutions, a tighter result of this form is derived in \cite{SanghamitraISIT2017}. Thus coded computing is provably useful in certain scaling regimes of $P$ and $N$.
\item[(b)] Size of the input is comparable or much larger in scaling sense than the overall computational complexity but the size of outputs is much smaller than computational complexity. Examples include matrix-vector products (Input Size, Computational Complexity$=\Theta(N^2)$, Output Size $=\Theta(N)$). For such computations, the encoding cost at the master node is \textbf{at least} of the same order as the input size since every symbol in the input has to be used in the encoding operation at least once. Thus, it is best if we are able to encode in advance and amortize costs over multiple computations. Decoding, on the other hand, can be performed online with negligible additional overhead if
$$ \text{Output Size} \times (g_d(P)+P) = o \left( \frac{\text{Computational Complexity}}{ P } \right).   $$
Thus coded computing is most useful when the input is known in advance and the encoding cost can be amortized.

\item[(c)] Size of the outputs is comparable to the overall computational complexity. One example is computing tensor products (Input Size $=\Theta(N^2)$, Computational Complexity, Output Size $=\Theta(N^4)$). For linear operations where the size of the output is comparable to the computational complexity, the decoding complexity could be as high as the computational complexity making it difficult to decode online. It would be interesting to find applications or regimes where coded computing strategies can still be useful for such problems. Interestingly, for tensor products, \textit{encoding} can be done online as the size of the input is smaller in scaling sense than the overall computational complexity.
\end{itemize}

\subsection{Fully-Decentralized Implementations}
It will also be useful to obtain fully decentralized realizations of coded computing techniques with no centralized master node. This often avoids a ``single source of failure'', particularly if the encoder or decoder are themselves prone to straggling or errors. We refer interested readers to \cite{jeongFFT,mds_codenet,DNNPaperISIT,yang2014can,Tay_Bel_68,Yaoqing2017ITTrans} for works on completely decentralized implementations.

\section*{Acknowledgments}
We thank Mayank Bakshi and Yaoqing Yang for helpful discussions. We acknowledge support from NSF CNS-1702694, CCF-1553248, CCF-1464336, and CCF-1350314. This work was also supported in part by Systems on Nanoscale Information fabriCs (SONIC), one of the six SRC STARnet Centers, sponsored by MARCO and DARPA.

\bibliographystyle{IEEEtran}
\bibliography{IEEEabrv,sample}

\appendices
\section{Proof of Theorem~\ref{thm:nmat}} \label{app:proof1}
We will first prove Lemmas \ref{usefuleven} and \ref{usefulodd} which provide properties of coefficients of products of polynomials. Using Lemmas \ref{usefuleven} and \ref{usefulodd}, we show Claims \ref{cl:nmat1} and \ref{cl:nmat2} which demonstrate that the product $\mathbf{C}$ is contained in a set of coefficients of the matrix polynomial  $\Pi_{i=1}^{\lceil\frac{n}{2}\rceil}p_{\mathbf{C}_{i}}(x^{m^{i-1}})$, where $p_{\mathbf{C}_{i}}(x)$ is as defined in  (\ref{eq:pc}), for $i \in \{1, \cdots, \lceil n/2 \rceil\}$. Finally, we provide a proof of Theorem \ref{thm:nmat} using Claims \ref{cl:nmat1} and \ref{cl:nmat2}.

\begin{lemma}\label{usefuleven}
    If $p(x)=\sum_{j=0}^{2d^{i-1}-2}p_jx^j$  is a polynomial with degree $2d^{i-1}-2$ for some $i\geq 2$, and $q(x)=\sum_{j=0}^{2d-2}q_jx^j$  is any other polynomial with degree $2d-2$, then  $p_{d^{i-1}-1}q_{d-1}$ is the coefficient of $x^{d^i-1}$ in $p(x)q(x^{d^{i-1}})$.
    \end{lemma}
   %\begin{remark}
  % Note that $p_i$ is the coefficient of $x^{i}$, but it does not count the number of coefficients. That is, due to alignment between different degrees in multiplication some degree may not appear in $p(x)q(x)$.
   %\end{remark}
\begin{proof}
We first expand out $p(x)$ and $q(x)$ as following:
\begin{align}p(x)=\underbrace{\sum_{j=0}^{d^{i-1}-2}p_jx^j}_{\tilde{p}_1(x)}+p_{d^{i-1}-1}x^{d^{i-1}-1}+\underbrace{\sum_{j=d^{i-1}}^{2d^{i-1}-2}p_jx^j}_{\tilde{p}_2(x)}\label{l1}\end{align}
\begin{align}q(x)=\underbrace{\sum_{j=0}^{d-2}q_jx^j}_{\tilde{q}_1(x)}+q_{d-1}x^{d-1}+\underbrace{\sum_{j=d}^{2d-2}q_jx^j}_{\tilde{q}_2(x)}\label{l2}\end{align}

We show that the term of degree $d^i-1$ in $p(x)q(x^{d^{i-1}})$ is only generated by multiplication of  the term of degree $d^{i-1}-1$ in $p(x)$ and the term of degree $d^{i-1}(d-1)$ in $q(x^{d^{i-1}})$. For this purpose, we  consider following terms:
\begin{itemize}
\item[1.] Consider the multiplication of two lowest degree  terms in $\tilde{p}_1(x)$ and $\tilde{q}_2(x^{d^{i-1}})$ of equations (\ref{l1}) and (\ref{l2}).
That is, $q_{d^i}x^{d^i} p_0=p_0q_{d^i}x^{d^i}$ which has higher degree in comparison to $x^{d^i-1}$. Consequently, the degree of any term in the multiplication of  $\tilde{p}_1(x)$ and $\tilde{q}_2(x^{d^{i-1}})$ will be strictly greater than $d^{i}-1$.
\item[2.] Consider the multiplication of two highest degree of terms in $\tilde{p}_2(x)$ and $\tilde{q}_1(x^{d^{i-1}})$ of equations (\ref{l1}) and (\ref{l2}),  $$q_{d-2}x^{d^{i-1}(d-2)}p_{2d^{i-1}-2}x^{2d^{i-1}-2}=q_{d-2}p_{2d^{i-1}-2}x^{d^{i}-2}$$ is less than $d^i-1$. Consequently, the degree of any term in the multiplication of  $\tilde{p}_2(x)$ and $\tilde{q}_1(x^{d^{i-1}})$ will be strictly less than $d^{i}-1$.
\item[3.] Since the degree of any term in the multiplication of  $\tilde{p}_2(x)$ and $\tilde{q}_1(x^{d^{i-1}})$ is strictly less than $d^{i}-1$,  and any term in $\tilde{p}_1(x)$ has degree less than the degree of any term in $\tilde{p}_2(x)$, we conclude that any term in the multiplication of $\tilde{p}_1(x)$ and $\tilde{q}_1(x^{d^{i-1}})$ has degree strictly less than $d^i-1$.
\item[4.] Since the degree of any term in the multiplication of  $\tilde{p}_1(x)$ and $\tilde{q}_2(x^{d^{i-1}})$ is strictly greater than $d^{i}-1$, and any term in $\tilde{p}_2(x)$ has degree larger than the degree of any term in $\tilde{p}_1(x)$, we conclude that any term in $\tilde{p}_2(x)\tilde{q}_2(x^{d^{i-1}})$ has degree strictly greater than $d^i-1$ which completes the proof.
\end{itemize}
\end{proof}

 \begin{lemma}\label{usefulodd}
    If $p(x)=\sum_{j=0}^{2d^{i-1}-2}p_jx^j$  is a polynomial with degree $2d^{i-1}-2$ for some $i\geq 2$, and $q(x)=\sum_{j=0}^{d-1}q_jx^j$ is any other polynomial with  degree $d-1$, then,  for $0\leq j\leq d-1$, $p_{d^{i-1}-1}q_{j}$ are the coefficients of $x^{(j+1)d^{i-1}-1}$ in $p(x)q(x^{d^{i-1}})$.
    \end{lemma}
    \begin{proof}
   First, we expand out $p(x)$ as in (\ref{l1}), and expand $q(x)$ as follows:
\begin{align}q(x)=\underbrace{\sum_{k=0}^{j-1}q_k x^k}_{\tilde{q}_1(x)}+q_{j}x^{j}+\underbrace{\sum_{k=j+1}^{d-1}q_jx^j}_{\tilde{q}_2(x)}\label{eql2}\end{align}
    In order to prove Lemma \ref{usefulodd}, we show that  $x^{(j+1)d^{i-1}-1}$ term in $p(x)q(x^{d^{i-1}})$ is produced solely by the multiplication of the term $p_{d^{i-1}-1} x^{d^{i-1}-1}$ in $p(x)$ with the term $q_j x^{j(d^{i-1})}$ in $q(x^{d^{i-1}})$. First, it is clear that the product of the term $p_{d^{i-1}-1} x^{d^{i-1}-1}$ in $p(x)$ with the term $q_j x^{j(d^{i-1})}$ in $q(x^{d^{i-1}})$ has degree ${(j+1)d^{i-1}-1}$. Thus, to complete the proof, we  show that no other terms in $p(x)$ produce $x^{(j+1)d^{i-1}-1}$ term when multiplied with any term in $q(x^{d^{i-1}})$. To do so, we consider the following terms:

    \begin{itemize}
\item[1.] Consider the multiplication of two lowest degree  terms in $\tilde{p}_1(x)$ and $\tilde{q}_2(x^{d^{i-1}})$ as defined in (\ref{l1}) and (\ref{eql2}).
That is, $ p_0q_{j+1}x^{(j+1)d^{i-1}}$ which has higher degree in comparison to $x^{(j+1)d^{i-1}-1}$. Consequently, the degree of any term in the multiplication of  $\tilde{p}_1(x)$ and $\tilde{q}_2(x^{d^{i-1}})$ will be strictly greater than $(j+1)d^{i-1}-1$.
\item[2.] Consider the multiplication of two highest degree of terms in $\tilde{p}_2(x)$ and $\tilde{q}_1(x^{d^{i-1}})$,  the product $$q_{j-1}x^{(j-1)d^{i-1}}p_{2d^{i}-2}x^{2d^{i-1}-2}=q_{j-1}p_{2d^{i}-2}x^{(j+1)d^{i-1}-2}$$ has degree less than $(j+1)d^{i-1}-1$. Consequently, the degree of any term in the multiplication of  $\tilde{p}_2(x)$ and $\tilde{q}_1(x^{d^{i-1}})$ will be strictly less than $(j+1)d^{i-1}-1$.
\item[3.]  Since the degree of any term in the multiplication of  $\tilde{p}_2(x)$ and $\tilde{q}_1(x^{d^{i-1}})$ is strictly less than $(j+1)d^{i-1}-1$, and the degree of any term in   $\tilde{p}_1(x)$ is less than the degree of any term in $\tilde{p}_2(x)$, we conclude  that any term in the multiplication of $\tilde{p}_1(x)$ and $\tilde{q}_1(x^{d^{i-1}})$ has degree strictly less than $(j+1)d^{i-1}-1$.
\item[4.] Since the degree of any term in the multiplication of  $\tilde{p}_1(x)$ and $\tilde{q}_2(x^{d^{i-1}})$ is strictly greater than $(j+1)d^{i-1}-1$, and  the degree of any term in   $\tilde{p}_2(x)$ is larger than the degree of any term in $\tilde{p}_1(x)$, we conclude that any term in $\tilde{p}_2(x)\tilde{q}_2(x^{d^{i-1}})$ has degree strictly greater than $(j+1)d^{i-1}-1$ which completes the proof.
    \end{itemize}
    \end{proof}

Now, we are able to state the following claims.

\begin{claim}\label{cl:nmat1}
   The coefficient of $x^{m^{\lfloor \frac{n}{2}\rfloor}-1}$ in $\prod_{i=1}^{\lfloor \frac{n}{2}\rfloor} p_{\mathbf{C}^{(i)}}(x^{m^{i-1}})$ is $\prod_{i=1}^{\lfloor \frac{n}{2} \rfloor} \mathbf{A}^{(i)} \mathbf{B}^{(i)}$, where, for $i \in \{1, \cdots, \lfloor \frac{n}{2}\rfloor\}$, $p_{\mathbf{C}^{(i)}}(x)$ is as defined  in (\ref{eq:pc}).
 \end{claim}
    \begin{proof}
    We prove  the claim iteratively.  Since $p_{\mathbf{C}^{(1)}}(x)$ has degree $2m^{i-1}-2$ with $i=2$, and $p_{\mathbf{C}^{(2)}}(x)$ has degree $2m-2$, we have, by Lemma \ref{usefuleven}, that  the coefficient of $x^{m^2-1}$  in $p_{\mathbf{C}^{(1)}}(x)p_{\mathbf{C}^{(2)}}(x^{m})$ is the product of the coefficient of $x^{m-1}$ in $p_{\mathbf{C}^{(1)}}(x)$  and the coefficient of $x^{m^2-m}$ in $p_{\mathbf{C}^{(2)}}(x^m)$. However,  from Remark \ref{rmk:nmat}, we already know that $\mathbf{A}^{(1)} \mathbf{B}^{(1)}$ is the coefficient of $x^{m-1}$ in $p_{\mathbf{C}^{(1)}}(x)$ and that  $\mathbf{A}^{(2)} \mathbf{B}^{(2)}$ is the coefficient of $x^{m^2-m}$ in $p_{\mathbf{C}^{(2)}}(x^m)$. Therefore, $\mathbf{A}^{(1)} \mathbf{B}^{(1)} \mathbf{A}^{(2)} \mathbf{B}^{(2)}$ is the coefficient of  $x^{m^2-1}$  in $p_{\mathbf{C}^{(1)}}(x)p_{\mathbf{C}^{(2)}}(x^{m})$.

    Similarly, consider the two polynomials $p'(x)=p_{\mathbf{C}^{(1)}}(x)p_{\mathbf{C}^{(2)}}(x^{m})$ and $p_{\mathbf{C}^{(3)}}(x)$. Notice that $p'(x)$ has degree $2m^{i-1}-2$ with $i=3$, and $p_{\mathbf{C}^{(3)}}(x)$ has degree $2m-2$, therefore, from Lemma \ref{usefuleven},   the coefficient of $x^{m^3-1}$  in $p'(x)p_{\mathbf{C}^{(3)}}(x^{m^2})$ is the product of the coefficient of $x^{m^2-1}$ in $p'(x)$  and the coefficient of $x^{m^3-m^2}$ in $p_{\mathbf{C}^{(3)}}(x^{m^2})$. However,  from the previous step, we already know that $\mathbf{A}^{(1)} \mathbf{B}^{(1)} \mathbf{A}^{(2)} \mathbf{B}^{(2)}$ is the coefficient of $x^{m^2-1}$ in $p'(x)$. In addition, from Remark \ref{rmk:nmat},  we already know that  $\mathbf{A}^{(3)} \mathbf{B}^{(3)}$ is the coefficient of $x^{m^3-m^2}$ in $p_{\mathbf{C}^{(3)}}(x^{m^2})$. Therefore, $\mathbf{A}^{(1)} \mathbf{B}^{(1)} \mathbf{A}^{(2)} \mathbf{B}^{(2)} \mathbf{A}^{(3)} \mathbf{B}^{(3)}$ is the coefficient of  $x^{m^3-1}$  in $p'(x)p_{\mathbf{C}^{(3)}}(x^{m^2})=$ $p_{\mathbf{C}^{(1)}}(x)p_{\mathbf{C}^{(2)}}(x^{m}) p_{\mathbf{C}^{(3)}}(x^{m^2})$.

    Repeating the same procedure, we conclude that   $\prod_{i=1}^{\lfloor \frac{n}{2} \rfloor} \mathbf{A}^{(i)} \mathbf{B}^{(i)}$ is the coefficient of $x^{m^{\lfloor \frac{n}{2}\rfloor}-1}$ in $\prod_{i=1}^{\lfloor \frac{n}{2}\rfloor} p_{\mathbf{C}^{(i)}}(x^{m^{i-1}})$.
    \end{proof}

\begin{claim}\label{cl:nmat2}
If $n \geq 3$ and odd, then, for any $j\in \{1, \cdots, m\}$, $\left(\prod_{i=1}^{\lfloor \frac{n}{2} \rfloor} \mathbf{A}^{(i)} \mathbf{B}^{(i)} \right) \mathbf{A}^{(\lceil\frac{n}{2}\rceil)}_{j}$ is the coefficient of $x^{j m^{\lfloor \frac{n}{2}\rfloor}-1}$ in $\prod_{i=1}^{\lceil \frac{n}{2}\rceil} p_{\mathbf{C}^{(i)}}(x^{m^{i-1}})$, where, for $i \in \{1, \cdots, \lceil \frac{n}{2}\rceil\}$, $p_{\mathbf{C}^{(i)}}(x)$ is as defined  in (\ref{eq:pc}).
\end{claim}
\begin{proof}
First, notice that since the degree of $p_{\mathbf{C}^{(i)}}(x)$ is $2m-2$ for all $i \in \{1, \cdots, \lfloor \frac{n}{2} \rfloor\}$, the degree of $\Pi_{i=1}^{\lfloor\frac{n}{2}\rfloor}p_{\mathbf{C}^{(i)}}(x^{m^{i-1}})$ is $(2m-2) \sum_{i=1}^{\lfloor \frac{n}{2} \rfloor} m^{i-1}= 2m^{\lfloor \frac{n}{2} \rfloor}-2$. In addition, the matrix polynomial  $p_{\mathbf{C}^{(\lceil \frac{n}{2} \rceil)}}(x)$ has degree $m-1$. Therefore, from Lemma \ref{usefulodd}, for $1\leq j \leq m$, the product of the coefficient of $x^{m^{\lfloor \frac{n}{2} \rfloor}-1}$ in $\Pi_{i=1}^{\lfloor\frac{n}{2}\rfloor}p_{\mathbf{C}^{(i)}}(x^{m^{i-1}})$ and the coefficient of $x^{j-1}$ in $p_{\mathbf{C}^{(\lceil \frac{n}{2} \rceil)}}(x)$ is the coefficient of $x^{j m^{\lfloor \frac{n}{2} \rfloor}-1}$ in $\Pi_{i=1}^{\lceil\frac{n}{2}\rceil} p_{\mathbf{C}^{(i)}}(x^{m^{i-1}})$. However, we already know, from Claim \ref{cl:nmat1}, that the  coefficient of $x^{m^{\lfloor \frac{n}{2} \rfloor}-1}$ in $\Pi_{i=1}^{\lfloor\frac{n}{2}\rfloor} p_{\mathbf{C}^{(i)}}(x^{m^{i-1}})$ is $\prod_{i=1}^{\lfloor \frac{n}{2} \rfloor} \mathbf{A}^{(i)} \mathbf{B}^{(i)}$, also, by definition, the coefficient of $x^{j-1}$ in $p_{\mathbf{C}^{(\lceil \frac{n}{2} \rceil)}}(x)$ is $\mathbf{A}^{(\lceil \frac{n}{2} \rceil)}_{j}$. Thus, $\left(\prod_{i=1}^{\lfloor \frac{n}{2} \rfloor} \mathbf{A}^{(i)} \mathbf{B}^{(i)}\right) \mathbf{A}^{(\lceil\frac{n}{2}\rceil)}_{j}$ is the coefficient of $x^{j m^{\lfloor \frac{n}{2}\rfloor}-1}$ in $\prod_{i=1}^{\lceil \frac{n}{2}\rceil} p_{\mathbf{C}^{(i)}}(x^{m^{i-1}})$.
\end{proof}

Now, we  prove Theorem \ref{thm:nmat}.
 \begin{proof} [Proof of Theorem \ref{thm:nmat}]
 To prove the theorem, it suffices to show that for Construction \ref{con:nmatcod}, the fusion node is able to construct $\mathbf{C}$ from any $2m^{n/2}-1$ worker nodes if $n$ is even or from any $(m+1)m^{\lfloor \frac{n}{2}\rfloor}-1$ if $n$ is odd.

 First, for the case in which $n$ is even, we need to compute $\mathbf{C}=\prod_{i=1}^{\frac{n}{2}}\mathbf{A}^{(i)}\mathbf{B}^{(i)}$. Notice, from Claim \ref{cl:nmat1}, that the desired matrix product $\mathbf{C}$ is the coefficient of $x^{m^{n/2}-1}$ in $\prod_{i=1}^{ \frac{n}{2}} p_{\mathbf{C}^{(i)}}(x^{m^{i-1}})$. Thus, it is sufficient to compute this coefficient at the fusion node as the computation output for successful computation.  Now, because the polynomial $\prod_{i=1}^{ \frac{n}{2}} p_{\mathbf{C}^{(i)}}(x^{m^{i-1}})$ has degree $2m^{n/2}-2$, evaluation of the polynomial at any $2m^{n/2}-1$ distinct points is sufficient to compute all of the coefficients of powers of $x$ in $\prod_{i=1}^{ \frac{n}{2}} p_{\mathbf{C}^{(i)}}(x^{m^{i-1}})$ using polynomial interpolation. This includes $\mathbf{C}$, the coefficient of $x^{m^{n/2}-1}$.

Now, for the case in which $n$ is odd, we need to compute $\mathbf{C}=\left(\prod_{i=1}^{\lfloor{\frac{n}{2}}\rfloor}\mathbf{ A}^{(i)} \mathbf{B}^{(i)} \right) \mathbf{A}^{(\lceil{\frac{n}{2}}\rceil)}$. First, notice that $\mathbf{C}$ is a concatenation of the matrices $\left(\prod_{i=1}^{\lfloor{\frac{n}{2}}\rfloor}\mathbf{A}^{(i)} \mathbf{B}^{(i)}\right) \mathbf{A}^{(\lceil{\frac{n}{2}}\rceil)}_{j}$, $j \in \{1, \cdots,m\}$ as follows:
\begin{align}\label{eq:conc}
  &\mathbf{C}=\left( \prod_{i=1}^{\lfloor{\frac{n}{2}}\rfloor}\mathbf{ A}^{(i)} \mathbf{B}^{(i)} \right) \mathbf{A}^{(\lceil{\frac{n}{2}}\rceil)}=\notag\\
  &\left[\left(\prod_{i=1}^{\lfloor{\frac{n}{2}}\rfloor}\mathbf{ A}^{(i)} \mathbf{B}^{(i)} \right) \mathbf{A}^{(\lceil{\frac{n}{2}}\rceil)}_{1} \hspace{2pt} \right| \hspace{2pt} \cdots \hspace{2pt} \left| \hspace{2pt} \left(\prod_{i=1}^{\lfloor{\frac{n}{2}}\rfloor}\mathbf{A}^{(i)} \mathbf{B}^{(i)} \right) \mathbf{A}^{(\lceil{\frac{n}{2}}\rceil)}_{m}\right].
\end{align}
From Claim \ref{cl:nmat2}, for all $j \in \{1, \cdots, m\}$, the product $\left(\prod_{i=1}^{\lfloor{\frac{n}{2}}\rfloor}\mathbf{A}^{(i)} \mathbf{B}^{(i)} \right) \mathbf{A}^{(\lceil{\frac{n}{2}}\rceil)}_{j}$ is the coefficient of $x^{j m^{\lfloor \frac{n}{2}\rfloor}-1}$ in $\prod_{i=1}^{ \lceil\frac{n}{2}\rceil} p_{\mathbf{C}^{(i)}}(x^{m^{i-1}})$. Thus, it is sufficient to compute these coefficients, for all  $j \in \{1, \cdots, m\}$, at the fusion node as the computation output for successful computation.  Now, because the polynomial $\prod_{i=1}^{ \lceil\frac{n}{2}\rceil} p_{\mathbf{C}^{(i)}}(x^{m^{i-1}})$ has degree $m^{\lfloor\frac{n}{2}\rfloor}(m+1)-2$, evaluation of the polynomial at any $m^{\lfloor\frac{n}{2}\rfloor}(m+1)-1$ distinct points is sufficient to compute all of the coefficients of powers of $x$ in $\prod_{i=1}^{\lceil \frac{n}{2}\rceil} p_{\mathbf{C}^{(i)}}(x^{m^{i-1}})$ using polynomial interpolation. This includes  the coefficients of $x^{j m^{\lfloor \frac{n}{2}\rfloor}-1}$, i.e. $\left(\prod_{i=1}^{\lfloor{\frac{n}{2}}\rfloor}\mathbf{ A}^{(i)} \mathbf{B}^{(i)} \right) \mathbf{A}^{(\lceil{\frac{n}{2}}\rceil)}_{j}$, for all $j \in \{1, \cdots, m\}$.
 %The next section has a complexity analysis that shows that master and fusion nodes  have a lower computational complexity as compared with the workers.
\end{proof}

\section{Proof of Theorem~\ref{thm:nmat2}} \label{app:proof2}

\begin{proof} [Proof of Theorem~\ref{thm:nmat2}]
    We will first show that the maximum degree of $p_{\mathbf{C}}(x)$ is  $k(n,s,t)-1$, and we will show that we can recover $\mathbf{C}$ from the polynomial $p_{\mathbf{C}}(x)$. In this proof, we will only show for even $n$ as proving for odd $n$ only has mechanical differences.

    Let us first rewrite $p_{\mathbf{C}}(x)$:
    \begin{equation}\label{eq:nmat_pcx}
    p_{\mathbf{C}}(x) = \sum_{\substack{i_1=1\cdots t, \cdots, i_n=1\cdots s\\ j_1=1\cdots s, \cdots, j_n=1\cdots t}} \mathbf{A}^{(1)}_{i_1, j_1} \mathbf{B}^{(1)}_{i_2, j_2} \cdots \mathbf{A}^{(n/2)}_{i_{n-1}, j_{n-1}} \mathbf{B}^{(n/2)}_{i_n, j_n} x^{i_1 + t (s-1+j_1-i_2)
    % + st (j_2-i_3)
    +\cdots + s^{n/2-1}t^{n/2}(s-1+j_{n-1} - i_n) + s^{n/2}t^{n/2} j_n}.
    \end{equation}

    We get the maximum degree when $i_1 = t-1, s-1+j_1-i_2 = 2s-2, \cdots, j_n = t-1$. Hence,
    \begin{align}
    \textnormal{maximum degree of } p_{\mathbf{C}}(x) &= t-1 + t(2s-2) + st(2t-2) + \cdots s^{n/2-1}t^{n/2} (2s-2) + s^{n/2}t^{n/2}(t-1) \\
    &= t -1 -2t + 2 s^{n/2} t^{n/2} + s^{n/2}t^{n/2+1} -s^{n/2}t^{n/2} \\
    &= s^{n/2}t^{n/2+1} + s^{n/2} t^{n/2} -t -1 \\
    &= k(n,s,t) -1.
    \end{align}

    We now want to show that the coefficient of $x^{d(n,i,j)}$ is $\textbf{C}_{i,j}$, the $(i,j)$-th entry of $\mathbf{C}$,   where $d(n,i,j)$ is defined as follows:
    \begin{align}
        %\hline\nonumber\\
        {d(n,i,j)=}\left\{ \begin{array}{ll}
            i + t(s-1) + \cdots + s^{n/2-1}t^{n/2}(s-1) + j (s^{n/2}t^{n/2}) & \text{if $n$ is even,} \\
            i + t(s-1) + \cdots + s^{(n-1)/2}t^{(n-1)/2}(t-1) + j (s^{(n-1)/2}t^{(n+1)/2}) &\text{if $n$ is odd}
            \end{array}\right.\label{eq:dij}
            %\\
           % \hline\nonumber
    \end{align}

    Note that $\mathbf{C}_{i,j} = \sum_{j_1, j_2, \cdots, j_{n-1}} \mathbf{A}^{(1)}_{i,j_1} \mathbf{B}^{(1)}_{j_1, j_2} \mathbf{A}^{(2)}_{j_2, j_3} \mathbf{B}^{(2)}_{j_3, j_4} \cdots \mathbf{A}^{(n/2)}_{j_{n-2}, j_{n-1}} \mathbf{B}^{(n/2)}_{j_{n-1}, j}$. Among the terms in the sum in (\ref{eq:nmat_pcx}), $\mathbf{C}_{i,j}$ is the sum of terms that are from the $i$-th row of the first matrix $\mathbf{A}^{(1)}$ and the $j$-th column on the last matrix $\mathbf{B}^{(n/2)}$, and that have the second index and the first index of two adjacent matrices matching, i.e., $j_1 = i_2$ and $j_2 = i_3$. We can now show that from (\ref{eq:nmat_pcx}), we obtain $x^{d(n,i,j)}$ terms only when these conditions are satisfied:
    \begin{enumerate}[label={\roman*})]
    \item $i_1 = i$ \label{cond1}
    \item $j_1 = i_2, \cdots, j_{n-1} = i_n$ \label{cond2}
    \item $j_n = j$ \label{cond3}
    \end{enumerate}

    Let $d$ be the degree of $x$ in (\ref{eq:nmat_pcx})
    \begin{equation}\label{eq:d_radix}
    d = i_1 + t (s-1+j_1-i_2)+\cdots + s^{n/2-1}t^{n/2}(s-1+j_{n-1} - i_n) + s^{n/2}t^{n/2} j_n,
    \end{equation}
    and let us use an alternative representation of $d$ as follows:
    \begin{equation*}
    d = d_0 + d_1 \cdot t +d_2 \cdot st + \cdots + d_{n+1} \cdot s^{n/2}t^{n/2+1}
    \end{equation*}
    where
    \begin{align*}
    d_0 &= d \mod t \\
    d_1 &= (d-d_0)/t \mod s \\
    d_2 &= (d-d_0-d_1\cdot t)/st \mod t \\
    \vdots \\
    d_{n+1} &= (d-d_0-d_1\cdot t - d_n\cdot s^{n/2}t^{n/2})/s^{n/2}t^{n/2+1} \mod s.
    \end{align*}
    We can think of this representation as a mixed radix system $\mathcal{D}$ with $n+2$ digits, $(d_0, d_1, \cdots, d_{n+1})$, which has an alternating radix $(t, s, t, s, \cdots, t, s)$.
    By substituting $d_0 = t-1, d_1 = s-1, \cdots, d_{n+1} = s-1$, we can confirm that  the biggest number we can represent with (\ref{eq:d_radix}) is $s^{n+1} t^{n+1} -1  > k(n,s,t)-1$. Also, from its construction, any number between $0$ and $s^{n+1} t^{n+1} -1$ can be uniquely determined by the pair $(d_0, d_1, \cdots, d_{n+1})$ (for more explanation, see Theorem 1 in \cite{fraenkel1985systems}). Hence, any $0 \le d \le k(n,s,t)-1$ can be uniquely represented with $(d_0, d_1, \cdots, d_{n+1})$.

    Now, note that using (\ref{eq:d_radix}), $d(n,i,j)$ can be represented as $(d_0 = i, d_1 = t-1, \cdots, d_{n-1} = s-1, d_{n} = j, d_{n+1} = 0)$. Let us examine the condition for $d$ to have such digits. Let $d$ be any integer between $0$ and $k(n,s,t)-1$. Then, $d_0 = d \mod t = i_1$. Hence, $i_1 = i$ to have $d_0 = i$. Also, $d_1$ can be written as $d_1 = (j_2 - i_1 + s-1) \textnormal{ mod } s$. Since $j_2-i_1$ varies from $-s+1$ to $s-1$, $d_1 = s-1$ only when $j_2 = i_1$. By repeating this argument, we can prove the condition \ref{cond2}. Lastly, $d_n = j_n + 1$ if $s-1+j_{n-1}-i_n \ge s$ and $d_n = j_n$ if $s-1+j_{n-1}-i_n < s$. From our previous conditions, $j_{n-1}-i_n + s-1 = s-1$. Hence,  we can see that $d_n = j$ only when $j_n =j$. This completes the proof for $n$ even.
    \end{proof}

\end{document}